\documentclass[journal]{IEEEtran}
\usepackage{stfloats}
\usepackage{amsmath,amssymb,amsfonts}
\usepackage{amsthm}
\usepackage{calligra}
\usepackage{algorithmicx}
\usepackage[linesnumbered,ruled,vlined]{algorithm2e}
\usepackage{bm}
\usepackage{cite}
\usepackage{hyperref}
\usepackage{enumerate}
\usepackage[table]{xcolor}
\usepackage{graphicx}
\usepackage{epstopdf}
\usepackage{cite,url,epsfig}
\usepackage{epstopdf}
\usepackage{caption}
\usepackage{comment}
\usepackage{booktabs}
\usepackage{subfigure}
\usepackage{booktabs}
\usepackage{xcolor}
\usepackage{colortbl}
\usepackage{url}
\newtheorem{definition}{Definition}

\newtheorem{lemma}{Lemma}
\newtheorem{problem}{Problem}
\newtheorem{remark}{Remark}

\begin{document}
	\title{Contract-Theoretic Semantic Abstraction for Covert Communication in LAENets}
	\author{
		Yansheng Liu, Jinbo Wen, Kun Zhu*, \textit{Member, IEEE}, Hui Li, Yang Zhang, Jiawen Kang, \textit{Senior Member, IEEE}
		
		\thanks{ 
			Y. Liu, K. Zhu, H. Li, and Y. Zhang are with the College of Computer Science and Technology, Nanjing University of Aeronautics and Astronautics, Nanjing, China (e-mails:
			\nolinkurl{yansheng@nuaa.edu.cn};
			\nolinkurl{zhukun@nuaa.edu.cn};
			\nolinkurl{huili_cs@nuaa.edu.cn};
			\nolinkurl{yangzhang@nuaa.edu.cn}).
			
			J. Wen is with the Department of Computer Science, City University of Hong Kong, Hong Kong SAR, China (e-mail: jinbo1608@163.com).
			
			J. Kang is with the School of Automation, Guangdong University of Technology, Guangzhou, China (e-mail: kavinkang@gdut.edu.cn).

			\textit{*Corresponding author: Kun Zhu}
			
		} 
	}
	
	\maketitle
	
	\begin{abstract}
		In Low-Altitude Economy Networks (LAENets), Unmanned Aerial Vehicles (UAVs) provide diverse low-altitude services to the Base Station (BS) by transmitting task-related information under stringent low-probability-of-detection constraints. 
		These constraints severely limit the available transmission resources, making it challenging to reliably deliver task-relevant information using conventional bit-level communication.
        Fortunately, Semantic Communication (SemCom) can be a promising solution by prioritizing task-relevant information over bit-level accuracy. 
		However, different levels of semantic abstraction inherently introduce different degrees of information loss and redundancy, which may either compromise task reliability or incur excessive transmission overhead if not properly controlled.
        Moreover, since the BS cannot directly observe or reliably verify the semantic processing capabilities and abstraction-dependent transmission costs of UAVs, information asymmetry naturally arises in SemCom service provision.
        To this end, we propose an incentive-aware semantic abstraction framework for covert communications in LAENets.
        Specifically, we adjust the semantic abstraction level at the UAV side to provide differentiated covert semantic services under fixed-dimensional transmission payloads, thereby enabling reliable task information delivery under stringent covert constraints.
        Regarding the issue of information asymmetry, we propose a contract-theoretic model and adopt Prospect Theory (PT) to capture risk-sensitive semantic abstraction decisions under uncertainty.
		Furthermore, we design a Regularized Diffusion-based Soft Actor-Critic (RDSAC) algorithm for optimal contract design under PT.
		This algorithm enhances contract design by introducing diffusion entropy regularization together with action entropy regularization.
        Numerical results show the effectiveness of the proposed framework and demonstrate that RDSAC achieves average reward improvements of $6.77$\% over the SAC algorithm and $31.78$\% over the proximal policy optimization algorithm.
	\end{abstract}
	\IEEEpeerreviewmaketitle
	
	\section{Introduction}
	Recently, Low-Altitude Economy Networks (LAENets) have attracted increasing attention as a critical communication paradigm for operation-critical low-altitude applications such as urban sensing, infrastructure inspection, and emergency response \cite{11271498, 11270843}. 
    In LAENets, Unmanned Aerial Vehicles (UAVs) act as aerial sensing and communication entities that transmit task-related information to legitimate receivers, and in certain regulation-constrained scenarios, operate under covert communication requirements to remain undetectable to adversarial monitors \cite{9382022,11270843}. 
    In such scenarios, the covert communication constraints fundamentally differentiate LAENets from conventional UAV communication systems. 
    Specifically, to satisfy low-probability-of-detection constraints, UAV transmissions in LAENets operate under stringent power and spectral limitations \cite{11059486,10447237}, which substantially reduce the effective communication capacity \cite{11271498}. 
    Consequently, such severely constrained communication resources are insufficient to support UAV tasks with increasing semantic complexity, including scene understanding, target recognition, and cooperative decision-making \cite{11059486}. 
    Moreover, traditional physical-layer covert communication techniques 
    treat transmitted information as bit streams and remain largely agnostic to task semantics \cite{11059486}. This limitation makes them inefficient in scenarios where system performance is determined by high-level task understanding rather than raw data fidelity \cite{11098629}.
	
    Semantic Communication (SemCom) has recently emerged as a promising paradigm that prioritizes the transmission of task-relevant information over bit-level accuracy \cite{11098629}. This paradigm is particularly appealing for LAENets with covert communication constraints, as it enables task execution under extremely constrained communication budgets. 
	Instead of transmitting raw sensory data, SemCom allows UAVs to convey abstracted and task-relevant information that directly supports downstream tasks such as perception, recognition, and decision-making \cite{11059486}. 
	Nevertheless, different semantic abstraction levels may lead to different trade-offs among semantic reconstruction quality, transmission overhead, and covertness \cite{10447237, 10912460}. 
    For example, when a UAV transmits an inspection image under a covert-feasible transmission budget, semantic abstraction can suppress task-irrelevant background before encoding to emphasize target-related information. 
    However, aggressive abstraction can significantly reduce irrelevant visual content and transmission burden, but may discard useful contextual information and degrade task performance \cite{10854543}. 
	In contrast, excessive abstraction may discard useful contextual cues and degrade task performance \cite{10854543}. 
    Therefore, semantic abstraction should balance task-critical content and contextual detail within a fixed representation budget \cite{10447237,10912460}.

    Although SemCom has demonstrated effectiveness in bandwidth-limited or noisy channels, semantic abstraction-aware covert communication in LAENets still poses fundamental challenges. 
	Firstly, existing studies predominantly optimize robustness or task accuracy with fixed or centrally designed semantic abstraction strategies \cite{11059486,10912460,10447237}. 
    Such formulations overlook the intrinsic uncertainty introduced by semantic abstraction and its interaction with stringent covert communication constraints.
	Secondly, directly enforcing semantic abstraction decisions requires the Base Station (BS) to know the operational and semantic characteristics of UAVs, such as sensing quality, energy availability, onboard computational resources, and abstraction-dependent service costs \cite{10841438}. 
    However, in practical LAENets, these characteristics are privately known by UAVs. In addition, the BS cannot independently verify all reports, and self-interested UAVs may misrepresent their characteristics to obtain more favorable rewards \cite{10638123}.
    This information asymmetry makes centralized semantic service control difficult.
	Third, the effective quality of covert SemCom depends on multiple coupled factors, including covert feasibility, packet-level delivery reliability, and abstraction-dependent semantic reconstruction quality. 
	A unified task-oriented service-quality indicator is therefore needed to support incentive design and contract optimization.
	
	To address the above challenges, in this paper, we develop an incentive-aware semantic abstraction control framework for covert communication in LAENets.
	Instead of treating semantic uncertainty as a directly measurable information-theoretic quantity, we adopt verifiable semantic abstraction levels as service-level indicators \cite{10854543}. 
	Specifically, UAVs process raw sensory observations with different abstraction levels before semantic encoding, thereby providing differentiated covert semantic services under fixed-dimensional transmission payloads. 
	To evaluate the effective service quality, we introduce a covert semantic effectiveness metric that jointly accounts for covert feasibility, packet delivery reliability, and abstraction-dependent semantic reconstruction quality. 
	Under information asymmetry, where UAV-specific capabilities and costs are privately known \cite{10841438}, we employ contract theory to induce UAV self-selection through service-quality and reward pairs \cite{10638123}. 
	Moreover, we incorporate Prospect Theory (PT) into BS utility modeling to capture risk-sensitive semantic decisions under uncertainty \cite{xie2025privacy}. 
	To address the resulting non-convex contract optimization problem, we propose a regularized diffusion-based policy optimization algorithm that generates stable contract decisions for abstraction-aware covert SemCom in dynamic LAENets with covert communication constraints \cite{wen2025hybridrag}.
	The main contributions of this paper are summarized as follows:
	\begin{itemize}
		\item \textbf{Incentive-aware Semantic Abstraction Framework in LAENets:} We establish an incentive-aware semantic abstraction framework for covert communication in LAENets. 
		By adopting verifiable semantic abstraction levels at the UAV side, the proposed framework enables differentiated task-oriented covert transmission while remaining compatible with fixed-dimensional semantic encoding and transmission architectures.
		To quantify the semantic service quality, we design a covert semantic effectiveness metric that jointly captures covert feasibility, packet delivery reliability, and abstraction-dependent semantic reconstruction quality. 
		This metric provides the basis for coordinating semantic service quality, transmission overhead, and covertness under task-oriented communication requirements.
		
		\item \textbf{Contract-based Incentive Mechanism under PT:} To address information asymmetry between the BS and heterogeneous UAVs in LAENets, we formulate a contract-theoretic model in which abstraction-dependent covert semantic effectiveness and rewards are jointly designed as contract items. 
		Furthermore, we incorporate PT to model the risk-sensitive utility evaluation of the BS under uncertain covert sensing tasks, making the proposed incentive mechanism more suitable for regulation-aware semantic service provisioning in LAENets.
		
		\item \textbf{Regularized Diffusion-based Contract Optimization:} We propose a Regularized Diffusion-based Soft Actor-Critic (RDSAC) algorithm to learn contract policies for covert semantic service optimization. 
        Specifically, we employ diffusion models as the policy to generate contract items through iterative denoising.
		To enhance policy performance, we integrate diffusion entropy regularization and action entropy regularization into the policy optimization objective. 
        The former improves policy diversity and exploration stability, while the latter mitigates premature convergence. 
		Numerical results show that the proposed RDSAC algorithm outperforms DRL baselines, achieving an average reward improvement of $6.77$\% over the SAC algorithm and $31.78$\% over the Proximal Policy Optimization (PPO) algorithm.
	\end{itemize}
	
	The remainder of the paper is organized as follows: In Section \ref{two}, we review the related works. Section \ref{three} introduces the proposed framework for LAENets. In Section \ref{four}, we formulate the contract-based incentive mechanism under PT. In Section \ref{five}, we present the architecture of the RDSAC algorithm. Section \ref{six} provides numerical results to evaluate the performance of the proposed framework and the algorithm. In Section \ref{seven}, we conclude the paper.

	\begin{figure}[t]
		\centering
		\vspace{-0.5em}
		\captionsetup{font=footnotesize}
		\includegraphics[width=0.8125\linewidth]{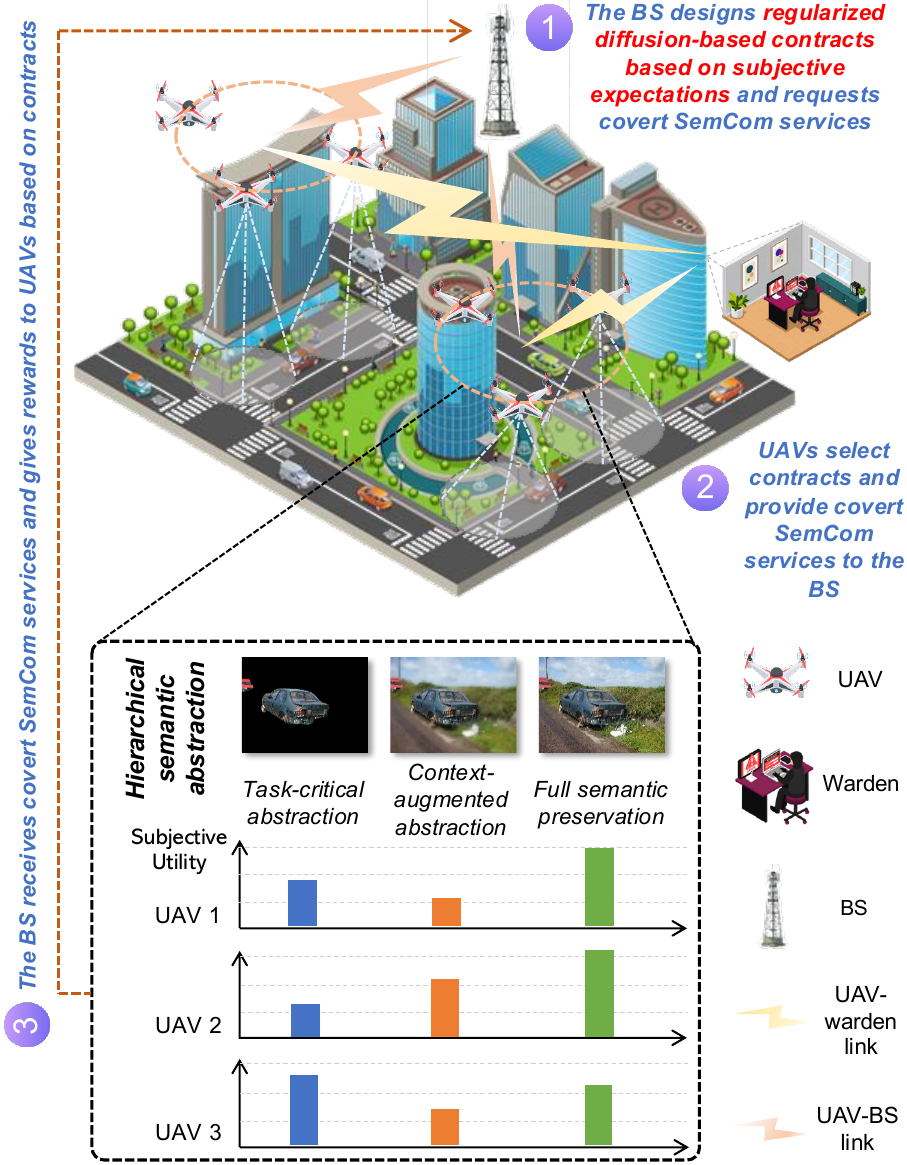}
		\caption{Incentive-aware semantic abstraction framework for covert communication in LAENets. The BS generates regularized diffusion-assisted contracts according to covert semantic communication requirements and PT-based utility evaluation. UAVs select contract items and adjust their semantic abstraction levels accordingly. Hierarchical semantic abstraction enables differentiated task-oriented semantic transmission under information asymmetry.}
		\label{system}
	\end{figure}
	\section{Related Work} \label{two}
	\subsection{Covert Communication and Secure Resource Management}
    In LAENets, covert communication has emerged as a prominent research focus in recent studies \cite{11271498,11270843,9382022,10447237}. An increasing number of researchers are dedicating their efforts to designing transmission and mobility strategies that satisfy low-probability-of-detection constraints, aiming to achieve reliable and undetectable UAV communications. For instance, the authors in \cite{11271498} investigated power control strategies for UAV-enabled covert communications under statistical detection constraints. The authors in \cite{11270843} analyzed the fundamental limits of covert transmission capacity in UAV-assisted wireless networks. The authors in \cite{9382022} further studied trajectory and resource optimization for UAV covert communication, exploiting UAV mobility to enhance covertness performance.
	In \cite{10447237}, the authors examined covert transmission schemes under fading channels and adversarial uncertainty. 

    Related studies have also researched network coordination and resource management in aerial and mobile networks. The authors in \cite{11097047} investigated environment-aware joint coverage and connectivity for integrated communication and jamming networks. In \cite{11614498}, the authors examined resource orchestration in 6G low-altitude IoT networks from a meta-learning perspective. The authors in \cite{11016100} studied the joint design of multi-UAV deployment, task offloading, and service placement for secure edge computing. In \cite{11592670}, the authors developed a UAV-assisted collaborative inference framework based on multi-exit Deep Neural Networks (DNNs) and cooperative jamming. They jointly optimized UAV trajectories, early-exit selection, DNN partitioning, and computation resource allocation.
	
    These studies advance covert transmission, network coordination, and secure computation in aerial networks. However, jointly accounting for semantic relevance and incentive compatibility introduces an additional challenge. For independently operated UAVs, semantic abstraction affects task utility and service costs, while relevant operating characteristics remain private. This motivates an incentive scheme that coordinates semantic abstraction under covertness constraints and information asymmetry.
	
	\subsection{Semantic Abstraction in SemCom}
    SemCom has attracted extensive attention as a means to improve communication efficiency by prioritizing task-relevant information over bit-level fidelity. By extracting high-level semantic representations, SemCom systems reduce transmission overhead while preserving task performance. 
	Existing studies have explored semantic abstraction as an effective approach to controlling the amount of task-relevant information conveyed under limited communication resources \cite{11098629,11059486,10447237}. 
	Specifically, the authors in \cite{11098629} developed a task-oriented SemCom framework based on DNNs to extract task-relevant features. 
	In \cite{11059486}, the authors adjusted the depth of semantic representations to improve efficiency in bandwidth-limited channels, while the authors in \cite{10447237} investigated diffusion-based semantic representation strategies, where different diffusion steps are employed to balance transmission cost and task performance. 
	In addition, the authors in \cite{11288062} proposed a task-oriented integrated sensing and SemCom framework for multi-device video analytics, which exploits channel state information to guide region-of-interest extraction and semantic transmission, thereby reducing both communication and device-side computational overhead.
	Moreover, the authors in \cite{11610885} proposed an intelligent semantic agent architecture. The architecture coordinates semantic perception, reasoning, and communication across cloud, edge, and end devices.
	Despite these advances, most existing studies rely on fixed or heuristically selected semantic abstraction strategies and usually assume that the semantic extraction process can be centrally designed or directly enforced. 
	They rarely consider how heterogeneous UAVs with private sensing, computation, and abstraction-dependent costs can be incentivized to provide differentiated covert semantic services under information asymmetry.
	\subsection{Diffusion-based Policy Optimization}
	Diffusion models have recently attracted increasing attention in policy learning and decision optimization due to their strong expressive capacity and intrinsic multimodality for modeling complex action distributions \cite{10834569,ding2024diffusion,du2024diffusion}.
	The authors in \cite{10834569} employed conditional diffusion models to parameterize Q-learning policies, enabling flexible action generation under complex state-action dependencies. 
	The authors in \cite{ding2024diffusion} proposed an efficient diffusion-based policy that reconstructs actions by reversing a progressive corruption process during training, thereby improving sampling efficiency and stability. 
	Building on these properties, diffusion-based Deep Reinforcement Learning (DRL) algorithms have been increasingly adopted for network optimization problems characterized by high dimensionality and non-convexity \cite{10638123,du2024diffusion}. 
	The authors in \cite{wen2025hybridrag} utilized diffusion-based DRL to derive carbon-aware computation offloading strategies, 
    while the authors in \cite{10638123} developed a diffusion-driven contract design policy to generate contracts in edge AIGC services. 
	Motivated by these studies, we adopt a regularized diffusion-based policy optimization algorithm to solve the contract optimization problem for abstraction-aware covert SemCom in LAENets. 
	\section{Framework Design} \label{three}
    As shown in Fig.~\ref{system}, we present an incentive-aware semantic abstraction framework for covert communication in LAENets.
	The considered system consists of a BS and $M$ heterogeneous UAVs, denoted by $\mathcal{M}=\{1,\dots, M\}$ \cite{10841438}.
	The BS publishes sensing tasks and offers rewards.
	Each UAV independently decides whether and how to participate based on its local operating conditions.
	Meanwhile, a warden, denoted by Willie, continuously monitors the wireless spectrum to detect transmissions from UAVs to the BS \cite{11059486}.
	Upon accepting a task, each UAV collects raw sensory data (e.g., images) from the target area.
	However, transmitting raw data is prohibitive due to the high risk of detection imposed by Willie.
	Therefore, the UAVs employ SemCom techniques to process the raw data, extracting semantic features that preserve the essential meaning while significantly reducing the transmission volume.

	\subsection{Covert Transmission Model}
	Due to the existence of Willie, we consider the impact of Willie before measuring the performance of UAVs. In real environments, many different objects act as scatterers or obstacles during Air-to-Ground (A2G) communications \cite{10134570}. Radio signals emitted by UAVs do not propagate in free space but may be affected by scattering or shadowing caused by objects \cite{10134570}, resulting in additional path loss. Thus, we capture the A2G communication by a probabilistic path loss model rather than a simplified free path loss model \cite{10381761, wen2025hybridrag, 10134570}. 
	
	Specifically, the adopted probabilistic path loss model accounts for both the occurrence probabilities and the corresponding path losses under Line-of-Sight (LoS) and Non-LoS (NLoS) conditions. Accordingly, the LoS and NLoS occurrence probabilities for the communication link between UAV $u_i \in \mathcal{M}$ and the BS at time slot $t$ are given by \cite{li2025covert}
	\begin{equation} \label{plos}
		\begin{aligned}
			& {\Pr}^{\mathrm{LoS}}_{u_i,\mathrm{BS}}(t) = \frac{1}{1 + a \exp\left(-b \left( \frac{180}{\pi} \theta_{u_i,\mathrm{BS}}(t) - a \right)\right)},  \\
			& {\Pr}^{\mathrm{NLoS}}_{u_i,\mathrm{BS}}(t) = 1 - {\Pr}^{\mathrm{LoS}}_{u_i,\mathrm{BS}}(t).
		\end{aligned}
	\end{equation}
	Here, $a$ and $b$ are constant parameters related to the environment, and $\theta_{u_i,\mathrm{BS}}(t) = \arctan \left(\frac{h_{u_i}}{d_{u_i,\mathrm{BS}}(t)} \right)$ represents the elevation angle of UAV $u_i$ to the BS, where $h_{u_i}$ is the flight altitude of UAV $u_i$ and $d_{u_i,\mathrm{BS}}(t)$ represents the horizontal distance between UAV $u_i$ and the BS at time slot $t$. Correspondingly, the path losses for the LoS and NLoS links are denoted by $h^{\mathrm{LoS}}(t)$ and $h^{\mathrm{NLoS}}(t)$, respectively, which are given by \cite{9112268}
	\begin{equation}
		\begin{aligned}
			h^{\mathrm{LoS}}(t) &= \left( \frac{4\pi f_c d_{u_i, \mathrm{BS}}^{\mathrm{3D}}(t)}{c} \right)^2 \beta^{\mathrm{LoS}}, \\
			h^{\mathrm{NLoS}}(t) &= \left( \frac{4\pi f_c d_{u_i, \mathrm{BS}}^{\mathrm{3D}}(t)}{c} \right)^2 \beta^{\mathrm{NLoS}}, 
		\end{aligned}
	\end{equation}
	where $f_c$ represents the carrier frequency, $c$ is the speed of light, and $\beta^{\mathrm{LoS}}$ and $\beta^{\mathrm{NLoS}}$ denote the attenuation coefficients corresponding to the LoS and NLoS links, respectively. $d_{u_i, \mathrm{BS}}^{\mathrm{3D}}(t)$ represents the three-dimensional distance between UAV $u_i$ and the BS at time slot $t$, which is given by 
	\begin{equation}
		d_{u_i, \mathrm{BS}}^{\mathrm{3D}}(t) = \sqrt{d_{u_i, \mathrm{BS}}^2(t) + h_{u_i}^2}.
	\end{equation}
	
	We denote the expected channel power gain as $\mathbb{E}[|h_{u_i,\mathrm{BS}}(t)|^2]$. Accounting for the probabilistic LoS and NLoS channel models, $\mathbb{E}[|h_{u_i,\mathrm{BS}}(t)|^2]$ can be derived as a weighted combination of the LoS and NLoS components, which is expressed as \cite{wei2025uav}
	\begin{equation}\label{h}
		\mathbb{E}\left[ |h_{u_i,\mathrm{BS}}(t)|^2 \right] = h^{\mathrm{LoS}}(t) {\Pr}^{\mathrm{LoS}}_{u_i,\mathrm{BS}}(t) + h^{\mathrm{NLoS}}(t) {\Pr}^{\mathrm{NLoS}}_{u_i,\mathrm{BS}}(t).
	\end{equation}
	
    To facilitate analytical evaluation, we consider an orthogonal uplink scheduling scheme, where only one UAV $u_i$ is scheduled to upload its semantic packet to the BS in each time slot $t$ \cite{10095748}. The other UAVs remain silent in the considered uplink band during that slot and therefore contribute no transmitted signal at either the BS or Willie. For simplicity, the time-slot index $t$ is omitted in the following detection model.

    The warden Willie attempts to determine whether the scheduled UAV $u_i$ is sending information to the BS. Regarding this, a binary hypothesis is considered where $\mathcal{H}_0$ indicates that no UAV signal is transmitted in the observed slot, and $\mathcal{H}_1$ indicates that the scheduled UAV $u_i$ is sending a signal. 
    Hence, the signal received by Willie under these two hypotheses is given by~\cite{li2025covert}
	\begin{equation}
		y_w =
		\begin{cases}
			n_w, & \mathcal{H}_0,\\
			\sqrt{P_{u_i}}\,h_{u_i,w}x_{u_i}+n_w,
			& \mathcal{H}_1,
		\end{cases}
	\end{equation}
    where $P_{u_i}$ denotes the transmit power of the scheduled UAV, $h_{u_i,w}$ denotes the UAV--Willie channel coefficient, and $x_{u_i}$ is the normalized transmitted symbol. 
	The noise $n_w$ denotes the Additive White Gaussian Noise (AWGN) with variance $\sigma_w^2$ and is independent of the transmitted signal.

    Willie measures the average received power. 
	Under the adopted mean-power approximation, the received power is expressed as \cite{11059486}
	\begin{equation} \label{power}
		\overline{P} \approx
		\begin{cases}
			\sigma_w^2, & \mathcal{H}_0,\\
			P_{u,w}+\sigma_w^2, & \mathcal{H}_1,
		\end{cases}
	\end{equation}
	where $P_{u,w}=P_{u_i}|h_{u_i,w}|^2$ is the signal power received from the scheduled UAV $u_i$.
	Willie makes a decision by comparing the measured power against a predefined detection threshold $\varsigma$: $\overline{P}\underset{\mathcal{D}_0}{\overset{\mathcal{D}_1}{\gtrless}}\varsigma$.
	Specifically, if the measured power $\overline{P}$ exceeds the threshold $\varsigma$, Willie infers the presence of UAV transmission; otherwise, it infers the absence of communication. In other words, the decision $\mathcal{D}_0$ is made if $\overline{P}\leq\varsigma$, and $\mathcal{D}_1$ otherwise. Hence, the total detection probability $\epsilon$ is given by \cite{9112268,wei2025uav}
	\begin{equation}
		\epsilon =
		\pi_1\cdot\Pr(\mathcal{D}_1\mid\mathcal{H}_1)
		+\pi_0\cdot\Pr(\mathcal{D}_0\mid\mathcal{H}_0),
	\end{equation}
	where $\pi_0$ and $\pi_1$ denote the prior probabilities of hypotheses $\mathcal{H}_0$ and $\mathcal{H}_1$, respectively. To maximize the uncertainty at Willie, we consider the classical strategy of equal prior probability $\pi_0 = \pi_1 = 0.5$, which has been proven in \cite{10841438} to minimize the detection probability. 
	

	For finite-sample detection analysis, Willie measures the average received power as 
	$\overline{P}=N_w^{-1}\sum_{n=1}^{N_w}|y_w[n]|^2$
	using $N_w$ independent zero-mean complex Gaussian observations.
	The channel power gain remains fixed over the detection window.
	The observation variances under $\mathcal{H}_0$ and $\mathcal{H}_1$ are $\sigma_w^2$ and $\sigma_w^2+P_{u,w}$, respectively.
	The false-alarm and missed-detection probabilities are defined as $P_{\mathrm{FA}}=\Pr(\mathcal{D}_1\mid\mathcal{H}_0)$ and $P_{\mathrm{MD}}=\Pr(\mathcal{D}_0\mid\mathcal{H}_1)$, respectively.
	Under equal priors, the correct-decision probability is given by $\epsilon=1-(P_{\mathrm{FA}}+P_{\mathrm{MD}})/2$. To characterize the optimal detection performance of Willie, the minimum total detection error probability is defined as \cite{11059486}
	\begin{equation} 
		\xi^\star(P_{u_i})
		=\min_{\varsigma\geq0}
		\left(P_{\mathrm{FA}}+P_{\mathrm{MD}}\right).
	\end{equation}

	A transmission is considered covert-feasible if $\xi^\star(P_{u_i})\geq1-\epsilon_c$, where $\epsilon_c\in(0,1)$ denotes the covertness tolerance.
	A smaller $\epsilon_c$ imposes a stricter covertness constraint.
	Accordingly, the covert-feasibility indicator $\epsilon_0$ is defined as \cite{8714018}
	\begin{equation}\label{covert} 
	\epsilon_0=
	\begin{cases}
	1, & \xi^\star(P_{u_i})\geq1-\epsilon_c,\\
	0, & \text{otherwise}.
	\end{cases}
	\end{equation}

    For the specified Willie detector and channel configuration, the maximum covert-feasible transmit power is defined as
	\begin{equation}
		P_{\mathrm{cov}}^{\max}(\epsilon_c)
		=
		\sup\left\{
		P_u\geq0:
		\xi^\star(P_u)\geq1-\epsilon_c
		\right\}.
	\end{equation}
	Under the adopted detection model, $\xi^\star(P_u)$ decreases with transmit power.
	When the boundary lies within the considered power interval, it is obtained by numerically solving $\xi^\star(P_{\mathrm{cov}}^{\max})=1-\epsilon_c$.


	\subsection{Covert Semantic Effectiveness Metric} 
	
    In LAENets, covertness constraints determine whether transmissions remain undetected, packet errors affect whether semantic representations are correctly received, and semantic abstraction controls the visual semantic content retained before encoding. 
    Motivated by these observations, we define a metric termed \emph{covert semantic effectiveness}, denoted by $\mathcal{Q}(g)$, which reflects how physical-layer effects and semantic reconstruction mechanisms jointly influence effective covert semantic delivery.
	In the following, we first characterize the Packet Error Rate (PER) over the covert UAV-BS link, and then introduce the semantic abstraction regulation module, which governs the degree of semantic abstraction prior to encoding.

	
	\subsubsection{Packet error rate}
    Under the orthogonal uplink scheduling scheme described above, only the scheduled UAV $u_i$ transmits in the considered uplink band in time slot $t$ \cite{10095748}. 
	The other UAVs therefore do not generate co-channel interference at the BS. Accordingly, the received Signal-to-Noise Ratio (SNR) at the BS for UAV $u_i$ is \cite{wei2025uav}
	\begin{equation}
		\Upsilon_{u_i,\mathrm{BS}}^t
		=
		\frac{P_{u_i}^t |h_{u_i,\mathrm{BS}}^t|^2}
		{\sigma_{\mathrm{BS}}^2},
	\end{equation}
	where $P_{u_i}^t$ denotes the transmit power of UAV $u_i$ at time slot $t$, $h_{u_i,\mathrm{BS}}^t$ is the channel coefficient between UAV $u_i$ and the BS, and $\sigma_{\mathrm{BS}}^2$ is the variance of the AWGN at the BS.
	
	For simplicity, the semantic information $\mathcal{S}_i$ is transmitted as a single packet in the uplink. The BS and the warden employ a Cyclic Redundancy Check (CRC) mechanism to verify the integrity of the received semantic information, where the CRC mechanism appends a checksum to the transmitted packet and declares a packet error if any inconsistency is detected at the receiver. For each $N$-bit semantic information $\mathcal{S}_i$, the PER at the BS is given by\cite{10095748}
	\begin{equation} \label{per}
		\mathcal{P}_{\mathrm{PER}_{u_i,\mathrm{BS}}}^t
		=
		1-
		\exp\left(
		-\frac{\ln(CN)}
		{\beta \Upsilon_{u_i,\mathrm{BS}}^t}
		\right)
		\Gamma\left(
		1+\frac{1}
		{\beta \Upsilon_{u_i,\mathrm{BS}}^t}
		\right),
	\end{equation}
	where $\beta$ and $C$ are the modulation-specific constants, and $\Gamma(\cdot)$ represents the standard Gamma function. In the considered system, error packets are directly discarded at the receiver, and the BS does not request retransmission of the corresponding semantic information from the UAVs \cite{10095748}.

	\subsubsection{Semantic abstraction module}
    Most image-oriented SemCom systems encode visual observations into fixed-dimensional latent representations for transmission \cite{9685667,10158995}. 
	However, a fixed feature dimension specifies the representation size but does not determine which visual content is retained before encoding. 
	To maintain the same transmission interface, we regulate the balance between task-critical regions and contextual details before semantic encoding. Specifically, we introduce a hierarchical semantic abstraction variable $g \in \{1,2,3\}$ to specify different levels of task-relevant semantic preservation. 
	UAV platforms leverage existing perception models, such as a lightweight object detector YOLOv$5$ and a semantic segmentation network DeepLabv$3$ with a ResNet backbone, to extract task-relevant semantic regions \cite{10935177}. Based on the output of the perception models, a binary semantic mask $\mathbf{M} \in \{0,1\}^{H \times W}$ is constructed, where $\mathbf{M}(\mathbf{p})=1$ indicates task-critical semantic regions, and $\mathbf{M}(\mathbf{p})=0$ corresponds to background content. Within this framework, we define three ordered semantic abstraction levels as follows:

	\begin{itemize}
		\item \textbf{Task-critical abstraction ($g = 1$):} 
		Only task-critical semantic regions are preserved, while task-irrelevant regions are suppressed. 
		This level retains only the content within task-critical regions while keeping the image dimensions and transmitted feature length unchanged:
		\begin{equation}
			\mathbf{X}^{(1)} = \mathbf{X} \odot \mathbf{M}.
		\end{equation}
		where $\odot$ denotes the element-wise product.
		
		\item \textbf{Context-augmented abstraction ($g = 2$):} 
		Task-critical regions are preserved without modification before encoding, while background regions are retained in a coarse form to provide contextual cues. 
		The background abstraction is implemented by a low-pass operator $\mathcal{B}_{\kappa}(\cdot)$, such as spatial downsampling followed by upsampling:
		\begin{equation}
			\mathbf{X}^{(2)} = 
			\left(\mathbf{X} \odot \mathbf{M}\right) 
			+ \left(\mathcal{B}_{\kappa}(\mathbf{X}) \odot (\mathbf{1}-\mathbf{M})\right).
		\end{equation}
		
		\item \textbf{Full semantic preservation ($g = 3$):} 
		The original observation is directly fed into the semantic encoder without abstraction, retaining the complete visual content before encoding:
		\begin{equation}
			\mathbf{X}^{(3)} = \mathbf{X}.
		\end{equation}
	\end{itemize}
	
    The abstracted observation $\mathbf{X}^{(g)}$ is then encoded, transmitted, and reconstructed as $\hat{\mathbf{X}}^{(g)}$ at the receiver. 
	The semantic abstraction level $g$ regulates the content presented to the encoder without changing the payload length or physical-layer transmission configuration. Specifically, for a given covertness tolerance, all abstraction levels are evaluated with the same transmit power, packet length, coding rate, and transmission duration. Under the adopted Gaussian observation model, the observation distributions at Willie are independent of $g$ when these settings are fixed. Therefore, all abstraction levels exhibit the same physical-layer detection performance. Their differences are reflected in the semantic reconstruction quality obtained from different retained content under the same covert-feasible transmission conditions.

	\subsubsection{Metric design}
	
	We first consider the physical-layer transmission reliability. According to (\ref{per}), the successful delivery probability of a semantic packet from UAV $u_i$ to the BS is
	\begin{equation}\label{perper}
		P_{\mathrm{del},u_i}=1-\mathcal{P}_{\mathrm{PER}_{u_i,\mathrm{BS}}}.
	\end{equation}
	In addition, covert communication imposes a transmission feasibility constraint, characterized by the covert-feasibility indicator $\epsilon_0$ in (\ref{covert}), which reflects whether the semantic transmission satisfies the required covertness condition.
	
	Upon successful and covert reception, the utility of the semantic packet is determined by the quality of the reconstructed semantic content. For a given semantic abstraction level $g$, the original observation $\mathbf{X}$ is processed by the semantic abstraction module, encoded and transmitted through the semantic transceiver, and reconstructed at the BS as $\hat{\mathbf{X}}^{(g)}$. Due to semantic abstraction and wireless transmission impairments, the reconstructed observation may suffer from task-relevant semantic distortion.
    To characterize this distortion, we define a Structural Similarity (SSIM)-induced semantic distortion score, which is given by \cite{11059486, 1284395}
	\begin{equation}
		D_S(g)=1-\mathrm{SSIM}(\mathbf{X},\hat{\mathbf{X}}^{(g)}),
	\end{equation}
	where $\mathrm{SSIM}(\mathbf{X},\hat{\mathbf{X}}^{(g)})$ measures the structural similarity between the original and reconstructed observations. 
	Since the original observation is used as the reference, $D_S(g)$ reflects both abstraction-induced content changes and reconstruction errors. Based on $D_S(g)$, we adopt an affine mapping to define a reconstruction-based semantic quality score as
	\begin{equation} \label{semanticut}
		I_S(\hat{\mathbf{X}}^{(g)})
		=1-P_{E,\mathrm{base}}(1+\lambda D_S(g)),
	\end{equation}
	where $P_{E,\mathrm{base}}$ denotes the baseline distortion parameter used in the quality model, and $\lambda$ controls its sensitivity to reconstruction distortion. For nonnegative $P_{E,\mathrm{base}}$ and $\lambda$, this mapping is non-increasing in $D_S(g)$. The score is clipped to $[0,1]$ to ensure boundedness. 
    $I_S$ provides a common reconstruction-based quality measure for contract evaluation, rather than a measured recognition accuracy. Higher global structural similarity does not necessarily imply better performance for every downstream sensing task \cite{10507167}. Therefore, task-specific recognition metrics may yield different rankings of the abstraction levels.
	
	Combining packet-level delivery reliability (\ref{perper}) and semantic utility (\ref{semanticut}), the expected semantic value delivered per transmission is given by~\cite{10095748}
	\begin{equation}
		V(g)=P_{\mathrm{del},u_i} \cdot I_S(\hat{\mathbf{X}}^{(g)}).
	\end{equation}
	Finally, considering the covertness requirement and the fixed transmission payload size, the covert semantic effectiveness metric is expressed as \cite{10841438}
	\begin{equation}
		\mathcal{Q}(g)=\alpha\cdot\epsilon_0\cdot\frac{V(g)}{L},
	\end{equation}
	where $\alpha>0$ is a scaling factor, and $L$ denotes the fixed payload length of the transmitted semantic representation. 
    The metric $\mathcal{Q}(g)$ is a reconstruction-based service-quality indicator for covert SemCom.
	Therefore, variations in $\mathcal{Q}(g)$ reflect the joint effects of covert feasibility, packet delivery reliability, and abstraction-dependent semantic reconstruction quality \footnote{Wireless uncertainties, such as transmit power variation, channel fading, and UAV mobility, are reflected in the covertness and packet delivery terms \cite{10090449}. More fine-grained uncertainty modeling can be incorporated into the proposed contract scheme without changing its basic structure.}.
	
    The above analysis characterizes the covert semantic effectiveness of different abstraction levels. 
    At the network level, semantic quality requirements are also coordinated with UAV participation incentives. 
    Although full semantic preservation may yield higher reconstruction quality, the BS evaluates this benefit together with the required rewards 
    to satisfy individual rationality 
    and incentive compatibility. 
    These incentive requirements depend on the privately known resource availability and service costs of individual UAVs. Thus, we introduce a contract-based mechanism to jointly design semantic quality requirements and rewards under information asymmetry \cite{10638123}.
	
	\section{Contract-based Incentive Mechanism Under Prospect Theory} \label{four}
	In this section, we present a contract-based incentive mechanism under PT to motivate UAVs to provide task-oriented covert SemCom services to the BS. We first formulate the utility functions of both UAVs and the BS. Then, we propose a contract-theoretic model and validate its feasibility.
	
	\subsection{Utility Function} 
    In the system model, the BS publishes sensing tasks, while each UAV provides task-oriented covert SemCom services by selecting an appropriate semantic abstraction level. 
    The energy constraints, onboard resource availability, concurrent workloads, and those associated with covert SemCom of individual UAVs constitute private information. 
    These characteristics affect resource allocation and participation decisions but cannot be directly observed or reliably verified by the BS. The resulting information asymmetry motivates the use of contract theory.
    As commonly adopted in contract-theoretic modeling, similar to \cite{xie2025privacy}, we abstract these factors into a single type parameter representing the preference of UAVs for covert SemCom. 
	The greater the type parameter of a UAV, the higher the covert semantic effectiveness $\mathcal{Q}$ it can achieve under the same contract requirement.
	For ease of analysis, the preference of UAVs is defined as a weighted value of battery capacity and data storage capacity, with each contributing half of the total weight \cite{10841438}.
    Thus, the UAVs can be classified into a set $\mathcal{K} = \{\theta_k : 1 \le k \le K \}$ of $K$ types. We denote the $k$-th type of UAVs as $\theta_k$. In nondecreasing order, the UAV types are sorted as $\theta_1 \le \theta_2 \le \cdot\cdot\cdot \le \theta_{K}$. A higher value of $\theta$ indicates a UAV with superior resources, capable of performing more effective semantic extraction and covert communication. To facilitate explanation, the UAV with type $k$ can be called the type-$k$ UAV, and the proportion of type-$k$ UAVs is denoted as $\lambda_k$, with $\sum^K_{k = 1} \lambda_k = 1$.
	
	\subsubsection{UAV utility}
    The utility of type-$k$ UAVs, denoted by $U_k^A$, is defined as the difference between the evaluated reward $r(\theta_k, R_k)$ and the modeled service cost $c(\mathcal{Q}_k)$ \cite{xie2025privacy}. 
    By adopting a linear service-cost model, the utility is \cite{11341901}
	\begin{equation}
		\begin{aligned}
			U_k^A(\mathcal{Q}_k, R_k) & = r(\theta_k, R_k) - c(\mathcal{Q}_k) \\
			& = \upsilon \theta_k R_k - a {\mathcal{Q}_k},
		\end{aligned}
	\end{equation}
	where $\upsilon$ is a reward scaling factor and $a$ is the effective unit service-cost coefficient. The term $a\mathcal{Q}_k$ models the aggregate economic cost of meeting the required semantic quality. It does not directly measure transmission energy or individual processing operations.

	\subsubsection{BS utility}
	The objective utility of the BS toward type-$k$ UAVs is denoted as $U_k^{\mathrm{BS}}$, which is the difference between the revenue generated by receiving the semantic performance $\mathcal{Q}_k$ and the reward $R_k$ paid to the type-$k$ UAV \cite{9220821}. The revenue is defined as a performance-based metric that evaluates $\mathcal{Q}_k$ \cite{10841438}. Thus, the utility is expressed as \cite{10271832}
	\begin{equation}
		U_k^{\mathrm{EUT}} = \beta \ln{(\mathcal{Q}_k + 1)} - R_k ,
	\end{equation}
	where $\beta > 0$ denotes the unit profit for $Q_k$.
	
	Under information asymmetry, the BS is aware of the type distribution but cannot observe the exact private type of each UAV \cite{10638123}, which introduces decision uncertainty. To address this issue, the BS can adopt Expected Utility Theory (EUT) to formulate its overall objective utility as
	\begin{equation}
		U_{\mathrm{EUT}}^{\mathrm{BS}} = M \sum_{k=1}^{K} \lambda_k U_k^{\mathrm{EUT}}.
	\end{equation}
	
    However, under uncertain and risky operating conditions, the decision-making process of the BS may deviate from the assumption of perfect rationality. 
	For example, when evaluating stochastic task completion gains or potential losses arising from failed or partially completed UAV sensing tasks, the BS may exhibit heterogeneous risk attitudes, including risk-seeking or risk-averse behaviors \cite{xie2025privacy}. 
	Therefore, EUT cannot fully capture the risk preferences of the BS during the uncertain decision-making process. 
	To address this limitation, we adopt PT to further model the objective utility of the BS, making the contract formulation more practical.
	PT evaluates outcomes relative to a reference point, distinguishing perceived gains from losses \cite{xie2025privacy}. 
	Given a reference point $U_\mathrm{ref}$ for all types of UAVs, we convert $U_k^\mathrm{EUT}$ into the subjective utility, which is given by \cite{10254627}
	\begin{equation}
		\label{PT}
		U_{k,\mathrm{PT}}^\mathrm{BS} = 
		\begin{cases}
			(U_k^\mathrm{EUT} - U_\mathrm{ref})^{\zeta^+}, &U_k^\mathrm{EUT} \ge U_\mathrm{ref}, \\
			-\eta (U_\mathrm{ref} - U_k^\mathrm{EUT})^{\zeta^-}, &U_k^\mathrm{EUT} < U_\mathrm{ref},
		\end{cases}
	\end{equation}
	where $\zeta^-, \zeta^+ \in (0, 1]$ denote the weighting parameters associated with gain and loss distortion, respectively, and $\eta \ge 0$ represents the loss aversion coefficient. Based on (\ref{PT}), the overall subjective utility of the BS is formulated as
	\begin{equation}
		U_\mathrm{PT}^\mathrm{BS} = M \sum_{k=1}^{K} \lambda_k U_{k,\mathrm{PT}}^{\mathrm{BS}}.
	\end{equation}
	
	\begin{figure*}[t] 
		\vspace{-1.8em}
		\centering
		\captionsetup{font=footnotesize}
		\includegraphics[width=0.8\textwidth]{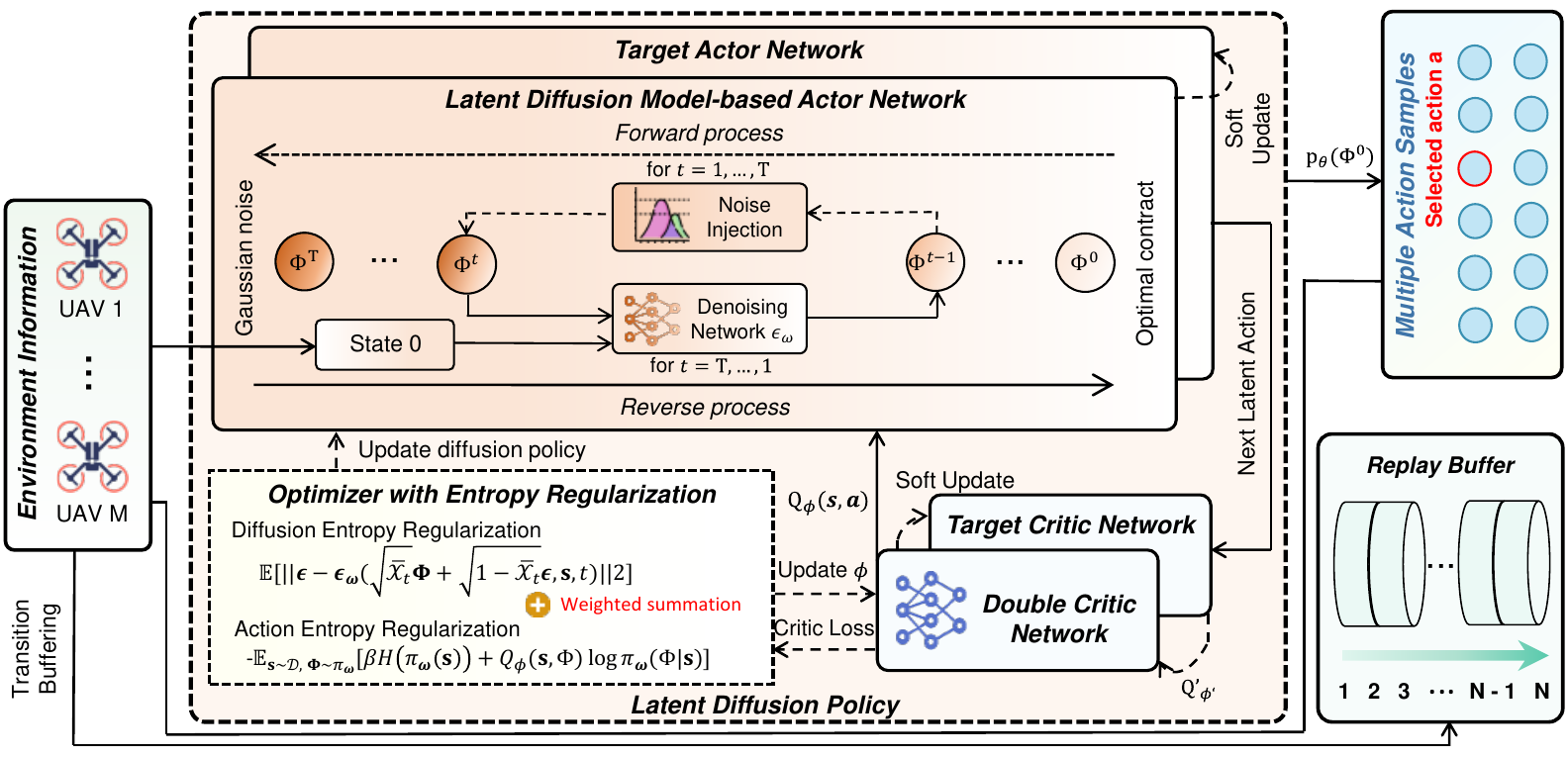}
		\caption{The architecture of the proposed RDSAC algorithm for optimal contract design under PT, where action entropy regularization and diffusion entropy regularization are incorporated into the diffusion-based policy to improve policy expressiveness and training stability.} 
		\label{framework}
	\end{figure*}
	
	\subsection{Contract Formulation} 
	To address information asymmetry between the BS and UAVs for task-oriented sensing and semantic information delivery, the BS acts as the leader and designs a menu of contract items, while each UAV selects the most suitable contract item according to its type. The contract item is denoted as $\Phi = {(\mathcal{Q}_k, R_k), k \in \mathcal{K}}$. To induce truthful selection of type-specific items, the feasible contract should satisfy the following Individual Rationality (IR) and Incentive Compatibility (IC) constraints.
	
	\begin{definition} 
		(Individual Rationality) The type-$k$ UAV achieves a nonnegative utility by selecting the contract item $(\mathcal{Q}_k, R_k)$ corresponding to its type, i.e.,
		\begin{equation} \label{ir}
			U_k^{A}(\mathcal{Q}_k, R_k) = \upsilon \theta_k R_k - a {\mathcal{Q}_k} \ge 0, \quad \forall k \in \mathcal{K}.
		\end{equation}
	\end{definition}
	
	\begin{definition}    
		(Incentive Compatibility) Each type-$k$ UAV prefers the contract item $(\mathcal{Q}_k, R_k)$ tailored to its type over any other item $(\mathcal{Q}_n, R_n), n \in \mathcal{K},$ and $n \neq k$, i.e.,
		\begin{equation}
			\label{IC}
			U_k^{A}(\mathcal{Q}_k, R_k) \ge U_k^{A}(\mathcal{Q}_n, R_n), \quad \forall n, k \in \mathcal{K}, n \neq k.
		\end{equation}
	\end{definition}
	
	The IR constraints can encourage UAVs to participate in covert SemCom by ensuring that the utility of UAVs is nonnegative, and the IC constraints can guarantee that each UAV provides high-quality covert SemCom by choosing the optimal contract item designed for its type. Based on the IR and IC constraints, the problem of maximizing the overall subjective utility of the BS is formulated as
	\begin{problem} \label{prob1}
		\begin{equation}
			\begin{aligned}
				& \qquad \max_{\mathbf{\mathcal{Q}}, \mathbf{R}} U^{\mathrm{BS}}_\mathrm{PT} \\
				& \text{s.t.} \quad U_k^{A}(\mathcal{Q}_k, R_k) \ge 0 , \forall k \in \mathcal{K}, \\
				& \phantom{\text{s.t.}} \quad U_k^{A}(\mathcal{Q}_k, R_k) \ge U_k^{A}(\mathcal{Q}_n, R_n), \forall n, k \in \mathcal{K}, n \neq k, \\
				& \phantom{\text{s.t.}} \quad \mathcal{Q}_k \geq 0, \, R_k \geq 0, \forall k \in \mathcal{K},
			\end{aligned}
		\end{equation}   
	\end{problem}
	where $\mathbf{\mathcal{Q}} = [\mathcal{Q}_k]_{1 \times K}$ and $\mathbf{R} = [R_k]_{1 \times K}$.
	\subsection{Optimal Contract Solution Feasibility}
	Since \textbf{Problem} \ref{prob1} contains $K(K-1)$ IC constraints, solving it directly becomes intractable. Therefore, we reformulate \textbf{Problem} \ref{prob1} based on the following necessary conditions.
	\begin{lemma}\label{lem6}
		With information asymmetry, a feasible contract should satisfy the following conditions:
		\begin{subequations}
			\begin{align} 
				& \upsilon \theta_1 R_1 - a \mathcal{Q}_1  \ge 0, \label{ir111} \\ 
				& \upsilon \theta_i R_i - a \mathcal{Q}_i \ge \upsilon \theta_i R_{i-1} - a \mathcal{Q}_{i-1}, \forall i \in \{2, \dots, K\}, \label{ic111}\\   
				& \upsilon \theta_i R_i - a \mathcal{Q}_i \ge \upsilon \theta_i R_{i+1} - a \mathcal{Q}_{i+1}, \forall i \in \{1, \dots, K - 1\}, \label{ic222}\\  
				& R_K \ge R_{K-1} \ge \cdot\cdot\cdot \ge R_1, \mathcal{Q}_K \ge \mathcal{Q}_{K-1} \ge \cdot\cdot\cdot \ge \mathcal{Q}_1.  \label{ic333}
			\end{align}
		\end{subequations}
	\end{lemma}
	\begin{proof}
		Please refer to \cite{10638123}.
	\end{proof}
	Constraint (\ref{ir111}) corresponds to the IR constraints, while (\ref{ic111}), (\ref{ic222}), and (\ref{ic333}) characterize the IC constraints. In particular, (\ref{ic111}) and (\ref{ic222}) indicate that the IC constraints can be equivalently reduced to the local downward incentive compatibility with monotonicity \cite{10638123}.
	
	Based on \textbf{Lemma} \ref{lem6}, we can derive the optimal reward to UAVs for covert SemCom service provisions by using the iterative method \cite{10271832}, given by
	\begin{equation} \label{reward}
		R_k^* = \frac{a}{\upsilon} \frac{\mathcal{Q}_1}{\theta_1} + \frac{a}{\upsilon} \sum_{i = 1}^k \Delta_i, \forall k \in \mathcal{K},
	\end{equation}
	where $\Delta_1 = 0$ and $\Delta_i = ((\mathcal{Q}_i - \mathcal{Q}_{i-1}) / \theta_i), i \in \{2, \dots, k\}$. By substituting the optimal reward expression in (\ref{reward}) into $\textbf{Problem}$~\ref{prob1}, the original optimization problem can be significantly simplified. Specifically, since the optimal rewards ${R_k^*}$ are fully determined by the variables ${\mathcal{Q}_k}$ and UAV type parameters ${\theta_k}$, the reward variables can be eliminated from the optimization. Therefore, $\textbf{Problem}$~\ref{prob1} is equivalently reduced to an optimization problem with respect to $\mathcal{Q}$ only, while all IC (\ref{IC}) and IR (\ref{ir}) constraints are implicitly satisfied. This reformulation substantially reduces the problem dimensionality and facilitates subsequent analysis and solution, which can be expressed as
	\begin{problem} \label{prob3}
		\begin{equation}  \label{prob33}
			\begin{aligned}
				& \qquad \max_{\mathbf{\mathcal{Q}}} U^{\mathrm{BS}}_\mathrm{PT} \\
				& \text{s.t.} \quad R_1 = \frac{a}{\upsilon} \frac{\mathcal{Q}_1}{\theta_1}, \\
				& \phantom{\text{s.t.}} \quad R_i = \frac{a}{\upsilon} \frac{\mathcal{Q}_1}{\theta_1} + \frac{a}{\upsilon} \sum_{i = 1}^i \frac{\mathcal{Q}_i - \mathcal{Q}_{i-1}}{\theta_i}, \forall i \in \{2, \dots, K\}, \\
				& \phantom{\text{s.t.}} \quad \mathcal{Q}_i \geq 0, \, R_i \geq 0, \forall i \in \{2, \dots, K\}.
			\end{aligned}
		\end{equation} 
	\end{problem}
	\begin{remark}
		We establish the existence and feasibility of an optimal contract in the proposed framework and model the PT-based contract design as a challenging decision-making problem under information asymmetry. Although conventional DRL methods such as PPO and SAC can be applied to contract optimization \cite{wang2022diffusion}, their reliance on local policy updates and explicit action sampling often leads to unstable convergence and inefficient exploration when the contract space is highly nonlinear and multimodal \cite{wen2025hybridrag,  ding2024diffusion}. In particular, under covert communication requirements, small perturbations in semantic abstraction decisions may induce non-smooth variations in the utility landscape, making standard DRL policies prone to premature convergence. To overcome these limitations, we adopt diffusion models as policy representations due to their ability to capture complex high-dimensional decision distributions \cite{ding2024diffusion}. Building on this, we propose the RDSAC algorithm, which jointly incorporates diffusion entropy regularization and action entropy regularization to balance exploration and stability while preserving the structural properties of the contract space \cite{10834569}, rendering diffusion-based DRL well-suited for solving the proposed PT-driven contract optimization problem.
	\end{remark}
	
	
	\section{Regularized-based Diffusion Models for Optimal Contract Design Under Prospect Theory} \label{five}

	In this section, we introduce the RDSAC algorithm for optimal contract design under PT, as illustrated in Fig.~\ref{framework}. We first present a DRL-based problem formulation for \textbf{Problem} \ref{prob3}, including the state space, action space, and reward function. Subsequently, we present the architecture of the proposed RDSAC algorithm.
	
	\subsection{DRL-based Problem Formulation}
	For each contract-design instance, the varying network and user parameters are independently sampled. The current contract decision does not affect the conditions of subsequent instances. 
	To support contract generation under varying conditions, we employ the DRL framework to learn a reusable state-conditioned policy offline. 
	The trained policy enables the BS to generate contracts through inference, reducing the need for repeated online optimization. 
	Specifically, the BS observes the current network condition and generates the corresponding contract items through a parameterized policy.
	The formulation is represented by $\langle \mathcal{S}, \mathcal{A}, \mathcal{R} \rangle$, where $\mathcal{S}$ denotes the state space composed of state vectors $\mathbf{s}$, $\mathcal{A}$ denotes the action space composed of action vectors $\mathbf{a}$, and $\mathcal{R}$ represents the reward function used to evaluate the quality and feasibility of the generated contract. 
	
	\subsubsection{State vector}
	In contract modeling, the environment that affects the optimal contract design is defined as
	\begin{equation}
		\mathbf{s} \triangleq \{M, K, U_\mathrm{ref}, a, \Upsilon, (\lambda_1, \dots, \lambda_K), (\theta_1, \dots, \theta_K) \},
	\end{equation}
	where the total dimensionality of the state $\mathbf{s}$ is ($2K + 5$).
	\subsubsection{Action vector}
	The objective of the BS is to determine an optimal contract conditional on the state $\mathbf{s}$. 
    Based on (\ref{reward}), the optimal reward can be derived from $\mathcal{Q}$. Hence, the action vector is given by
	\begin{equation}
		\mathbf{a} \triangleq \{\mathcal{Q}_1, \dots, \mathcal{Q}_K \},
	\end{equation}
	where the total dimensionality of the action $\mathbf{a}$ is $K$.
	\subsubsection{Reward function}
	The reward function typically accounts for both the objective and the associated constraints of  \textbf{Problem} \ref{prob3}. Thus, if the action $\mathbf{a}$ both satisfies the IC and IR constraints, $r(\mathbf{s}, \mathbf{a})$ is represented as
	\begin{equation}
		\begin{aligned}
			\mathcal{R} = r(\mathbf{s}, \mathbf{a}) &= U^{\mathrm{BS}}_\mathrm{PT} + \sum_{k = 1}^K U_k^A(\mathcal{Q}_k,R_k) \\ 
			& + \sum_{k = 1}^K\sum_{n = 1}^K (U_k^A(\mathcal{Q}_k,R_k) - U_k^A(\mathcal{Q}_n,R_n)),
		\end{aligned}
	\end{equation}
	otherwise, $r(\mathbf{s}, \mathbf{a})$ is set to $0$.
	\subsection{Algorithm Architecture}
	\subsubsection{Contract generation policy}
	As illustrated in Fig.~\ref{framework}, the contract generation policy specifies how the BS determines an appropriate contract design under a given environmental state. In this work, the policy is parameterized by a conditional diffusion model \cite{10834569}, which maps the observed state $\mathbf{s}$ to an optimal contract action that maximizes the expected cumulative reward over a series of time steps. Specifically, we denote the diffusion-based contract generation policy as $\pi_\omega(\Phi \mid \mathbf{s})$, where $\Phi=(\mathcal{Q}, R)$ denotes the contract pair and $\omega$ denotes the policy parameters.
	
	The proposed policy generates near-optimal or optimal contracts through the reverse denoising process of a conditional diffusion probabilistic model \cite{ding2024diffusion}. Starting from a Gaussian noise sample $\Phi^T \sim \mathcal{N}(\mathbf{0}, \mathbf{I})$, the policy progressively refines the noisy contract variable over $T$ diffusion steps conditioned on the current environment state $\mathbf{s}$ \cite{wang2022diffusion}. The overall policy distribution is given by \cite{du2024diffusion}
	\begin{equation}
		\begin{aligned}
			\pi_\omega(\Phi \mid \mathbf{s}) &= p_\omega(\Phi^0, \ldots, \Phi^T \mid \mathbf{s}) \\
			&= \mathcal{N}(\Phi^T;\mathbf{0},\mathbf{I}) \prod_{t=1}^{T} p_\omega(\Phi^{t-1}\mid \Phi^{t}, \mathbf{s}),
		\end{aligned}
	\end{equation}
	and the outcome of the reverse chain process represents the selected contract design. Each reverse transition is modeled as a Gaussian distribution, given by \cite{ding2024diffusion}
	\begin{equation}\label{phi0}
		p_\omega(\Phi^{t-1}\mid \Phi^{t}, \mathbf{s}) = \mathcal{N} \big(\Phi^{t-1}; \boldsymbol{\mu}_\omega(\Phi^{t},\mathbf{s},t), \boldsymbol{\Sigma}_\omega(\Phi^{t},\mathbf{s},t) \big), 
	\end{equation}
	where the covariance matrix is fixed as $\boldsymbol{\Sigma}_\omega(\Phi^{t},\mathbf{s},t) = \delta_t \mathbf{I}$. The mean $\boldsymbol{\mu}_\omega(\Phi^{t},\mathbf{s},t)$ is expressed as \cite{10638123}
	\begin{equation} \label{mean}
		\boldsymbol{\mu}_\omega(\Phi^{t},\mathbf{s},t)=\frac{1}{\sqrt{\chi_t}}\left(\Phi^{t}-\frac{\delta_t}{\sqrt{1-\bar{\chi}_t}}\boldsymbol{\epsilon}_\omega(\Phi^{t},\mathbf{s},t)\right),
	\end{equation}
	where $\chi_t = 1 - \delta_t$ and $\bar{\chi}_t = \prod_{i=1}^{t} \chi_i$. Here, $\boldsymbol{\epsilon}_\omega(\cdot)$ represents the contract generation network that predicts the denoising direction conditioned on the state $\mathbf{s}$ and the diffusion step $t$. The output of the reverse diffusion chain, i.e., $\Phi^0$, corresponds to the final contract design executed in the environment.
	\subsubsection{Contract quality network}
	The effective training of the contract design policy $\pi_\omega$ in dynamic and high-dimensional environments $\mathbf{s}$ can facilitate the training of the contract generation network $\epsilon_\omega$. To evaluate the long-term effectiveness of a given contract design under stochastic environments, we introduce a contract quality network, denoted as $Q_\phi(\mathbf{s},\Phi)$, which assigns a scalar value to each state–contract pair \cite{du2024diffusion}. This value represents the expected cumulative reward obtained by executing the contract $\Phi$ from state $\mathbf{s}$ and thereafter following the current contract generation policy \cite{wang2022diffusion}.
	
	The contract quality network plays a role analogous to the Q-function in DRL. Formally, under the maximum-entropy RL paradigm, the optimal policy is given by \cite{ding2024diffusion} 
	\begin{equation} \label{policy_update}
		\pi^\star=\arg\max_{\pi}\mathbb{E}\left[\sum_{z=0}^{\infty}\gamma^z\big(r(\mathbf{s}_z,\Phi_z)-\varsigma \log \pi(\Phi_z\mid \mathbf{s}_z)\big)\right],
	\end{equation}
	where $r(\mathbf{s}_z, \Phi_z)$ denotes the immediate reward determined by the PT utility of the BS and the contract feasibility constraints, $\gamma$ is the discount factor, and $\varsigma$ is the entropy coefficient.
	
	To improve training stability, we adopt double Q-learning and maintain two contract quality networks $Q_{\Phi_1}$ and $Q_{\Phi_2}$.  The critic parameters are optimized by minimizing the Bellman operator, given by \cite{wang2022diffusion}
	\begin{equation} \label{Q_update}
		\begin{split}
			\mathbb{E}_{(\mathbf{s},\Phi,r,\mathbf{s}')\sim\mathcal{D}} \Bigg[ \sum_{m=1}^{2} \big( & r+\gamma (1-d)\bar{Q}_{\bar{\phi}}(\mathbf{s}',\Phi') \\
			& -Q_{\phi_m}(\mathbf{s},\Phi)\big)^2 \Bigg],
		\end{split}
	\end{equation}
	where $\mathcal{D}$ denotes the experience replay buffer, $d$ is the termination indicator, $\Phi' \sim \pi_\omega(\cdot \mid \mathbf{s}')$ indicates that the contract action, and $\bar{Q}_{\bar{\phi}}$ is the target contract quality network. 
	\subsubsection{Policy improvement module}
	\begin{algorithm}[ht] 
		\DontPrintSemicolon
		\SetAlgoLined
		\caption{RDSAC Algorithm for Optimal Contract Design under PT}\label{AI_Contract}
		
		\KwIn{Diffusion steps $T$, interaction steps per episode $N$, mini-batch size $O$, discount factor $\gamma$, and soft target update coefficient $\tau$.}
		\KwOut{The optimal contract design $\Phi^0$.}
		Initialize replay buffer $\mathcal{D}$, contract generation network $\boldsymbol{\epsilon}_\omega$ with weights $\omega$, target contract generation network $\boldsymbol{\epsilon}'_{\omega'}$ with weights $\omega'$, contract quality network $Q_\phi$ with weights $\phi$, target contract quality network $Q'_{\phi'}$ with weights $\phi'$.\\
		
		\For{\rm{Episode} $e=1$ \rm{to} $E_\mathrm{max}$}
		{
			Initialize a random process $\mathcal{N}$ for contract design exploration. \\
			\For{\rm{Step} $n=1$ \rm{to} $N$}
			{
				Observe state $\mathbf{s}$ and randomly initialize a normal sample $\Phi^T_n \sim \mathcal{N}(\mathbf{0,I})$. \\
				\For{ $t=1$ \rm{to} $T$}
				{
					Construct a denoising network $\boldsymbol{\epsilon}_\omega(\Phi^t_n, \mathbf{s}, t)$.\\
					Calculate the mean by using (\ref{mean}).\\
					Generate the contract $\Phi^0_n$ based on (\ref{phi0}).\\
				}
				Execute the generated contract $\Phi_n^0$ and observe the reward $r_n$. \\
				Store record $(\boldsymbol{s}_n,\Phi_n^0, r_n, \boldsymbol{s}_{n+1})$ into the replay buffer $\mathcal{D}$. \\
			}
			
			Sample a random mini-batch of transitions $\mathcal{O}$ with size $O$ from $\mathcal{D}$. \\ 
			Update the contract quality network by minimizing the Bellman operator (\ref{Q_update}). \\
			Update the contract generation network by computing the policy gradient (\ref{policy_update}).\\
			Update the target networks:
			$\omega'\leftarrow\tau\omega+(1-\tau)\omega'$,
			$\phi'\leftarrow\tau \phi+(1-\tau)\phi'$. \\
		}
		\textbf{return} The trained contract generation network $\boldsymbol{\epsilon}_\omega$. \\
	\end{algorithm}
	Directly optimizing the diffusion-based contract generation policy under the above objective is nontrivial in large-scale multi-UAV systems. In particular, obtaining expert demonstrations with consistently high-quality state–action pairs and reliable $Q$ values is typically impractical \cite{du2024diffusion}. Without sufficient expert guidance, the policy may suffer from biased action exploration, leading to premature convergence to suboptimal solutions \cite{ding2024diffusion}.
	
	To address the above challenge, we reformulate the policy learning objective function by incorporating action entropy regularization and diffusion entropy regularization. Following the policy gradient formulation of maximum-entropy RL, maximizing the expected Q-value is equivalent to minimizing a weighted log-likelihood objective \cite{du2024diffusion}. Thus, the action entropy regularization loss is expressed as 
	\begin{equation} \label{acten}
		\begin{split}
			\mathcal{L}_{\mathrm{act}}(\omega) = -\mathbb{E}_{\mathbf{s}\sim\mathcal{D},\,\Phi\sim\pi_\omega} \big[ & Q_\phi(\mathbf{s},\Phi)\log \pi_\omega(\Phi\mid \mathbf{s}) \\
			& + \beta H(\pi_\omega(\mathbf{s})) \big],
		\end{split}
	\end{equation}
	where $Q_\phi(\mathbf{s},\Phi)=\min\{Q_{\phi_1}, Q_{\phi_2}\}$ and ${H}(\pi_\omega(\mathbf{s}))$ denotes the policy entropy. The action entropy regularization encourages the policy to maintain a more uniform action distribution, thereby enhancing exploration during training. Thus, the action entropy regularization in (\ref{acten}) effectively mitigates premature convergence to suboptimal solutions \cite{wen2025hybridrag}.
	
	Since the diffusion-based policy does not admit an explicit closed-form density, directly computing $\log \pi_\omega(\Phi | \mathbf{s})$ is intractable. To this end, we introduce a diffusion policy regularization term inspired by denoising diffusion probabilistic models \cite{wang2022diffusion}. This regularization minimizes the discrepancy between the true noise and the predicted noise during the forward diffusion process, which is expressed as \cite{wang2022diffusion}
	\begin{equation}
		\mathcal{L}_{\mathrm{diff}}(\omega)=\mathbb{E}\left[\left\|\boldsymbol{\epsilon}-\boldsymbol{\epsilon}_\omega\big(\sqrt{\bar{\chi}_t}\Phi+\sqrt{1-\bar{\chi}_t}\boldsymbol{\epsilon},\mathbf{s},t\big)\right\|^2\right],
	\end{equation}
	where $\boldsymbol{\epsilon}\sim\mathcal{N}(\mathbf{0},\mathbf{I})$ and $\sqrt{\bar{\chi}_t}\Phi+\sqrt{1-\bar{\chi}_t}$ represents the expert action after the reverse denoising process. This term can be interpreted as a behavior-cloning loss.
	
	Therefore, the contract generation policy is optimized by minimizing a weighted combination of the two losses, which is given by 
	\begin{equation}
		\mathcal{L}(\omega)=\rho\,\mathcal{L}_{\mathrm{act}}(\omega)+(1-\rho)\,\mathcal{L}_{\mathrm{diff}}(\omega),
	\end{equation}
	where $\rho\in[0, 1]$ balances value-driven optimization and diffusion-based regularization.
	\subsubsection{Complexity analysis}
	The proposed RDSAC algorithm involves utilizing denoising techniques to generate the optimal contract under PT \cite{10638123}, as illustrated in Algorithm~\ref{AI_Contract}, which facilitates the derivation of an optimal contract that maximizes the overall subjective utility of the BS under PT. 
    
    Algorithm~\ref{AI_Contract} consists of four phases. 
	In the algorithm initialization stage, the algorithm incurs a computational complexity of $\mathcal{O}(4|\omega| + 2|\phi|)$. The action sampling stage requires $\mathcal{O}(E_\mathrm{max}NT|\omega|)$ operations \cite{du2024diffusion}. Let $V$ denote the interaction cost between the DRL agent and the environment. The experience collection stage therefore exhibits a complexity of $\mathcal{O}(E_\mathrm{max}NV)$ \cite{du2024diffusion}. During parameter updates, the overall complexity comprises three components: $\mathcal{O}(OE_\mathrm{max}|\omega|)$ for policy updates, $\mathcal{O}(OE_\mathrm{max}|\phi|)$ for critic updates, and $\mathcal{O}(OE_\mathrm{max}(|\omega| + |\phi|))$ for target network updates. Consequently, the total complexity of the parameter update stage is $\mathcal{O}((O + 1)E_\mathrm{max}(|\omega| + |\phi|))$ \cite{wen2025hybridrag}.
	Based on the above analysis, the computational complexity of the proposed algorithm is $\mathcal{O}(4|\omega| + 2|\phi| + E_\mathrm{max}N(T|\omega| + V) + E_\mathrm{max}(O + 1)(|\omega| + |\phi|))$.

    Regarding deployment timeliness, RDSAC separates offline policy training from online contract generation. The trained policy generates contracts through $T$ diffusion denoising steps, with the corresponding rewards computed using (\ref{reward}). No additional training interactions or network updates are required during online inference. Under the above computational model, contract generation has a complexity of $\mathcal{O}(T|\omega|)$ per decision. Contracts can be generated at task initiation or when relevant operating conditions change. They can then be reused while these conditions remain unchanged, reducing repeated computation during task execution. Online inference latency depends on the deployment hardware, network size, and number of denoising steps. For low-altitude sensing applications, contract generation and dissemination should be completed within the task update interval. Overall, offline policy training and contract reuse reduce online computational overhead and can support timely task-level incentive decisions in low-altitude sensing applications.

    \begin{table}[!t]
    \centering
    \small
    \renewcommand{\arraystretch}{1.15}
    \caption{Main network and training hyperparameters of the proposed RDSAC.}
    \label{tab:rdsac_hyperparameters}

    \small
    \renewcommand{\arraystretch}{1.25}
    \setlength{\tabcolsep}{10pt}
    \setlength{\heavyrulewidth}{0.9pt}
    \setlength{\lightrulewidth}{0.5pt}

    \begingroup
    \rowcolors{2}{gray!5}{white}
    \begin{tabular}{lc}
        \toprule
        \rowcolor{gray!12}
        \textbf{Hyperparameters} & \textbf{Setting} \\
        \midrule
        Hidden-layer sizes & $[512,512]$ \\
        Optimizer & Adam \\
        Actor learning rate & $2\times10^{-7}$ \\
        Critic learning rate & $2\times10^{-6}$ \\
        Weight decay & $10^{-4}$ \\
        Discount factor $\gamma$ & $0.99$ \\
        Soft target-update coefficient $\tau$ & $0.005$ \\
        Mini-batch size & $512$ \\
        Replay-buffer capacity & $10^6$ \\
        Action-entropy coefficient & $0.05$ \\
        Training epochs & $200$ \\
        Environment interactions per epoch & $1000$ \\
        Decision steps per episode & $200$ \\
        Random seeds & $1,\ 312,\ 3407$ \\
        \bottomrule
    \end{tabular}
    \endgroup
    \end{table}
	
	\section{Numerical Results} \label{six}
	In this section, we first introduce the experimental setup. We then evaluate the proposed covert semantic abstraction framework. Finally, we validate the effectiveness of the proposed RDSAC algorithm for optimal contract design under PT. 
    \begin{figure*}[t] 
		\vspace{-1.8em}
		\centering
		\captionsetup{font=footnotesize}
		\includegraphics[width=0.78\textwidth]{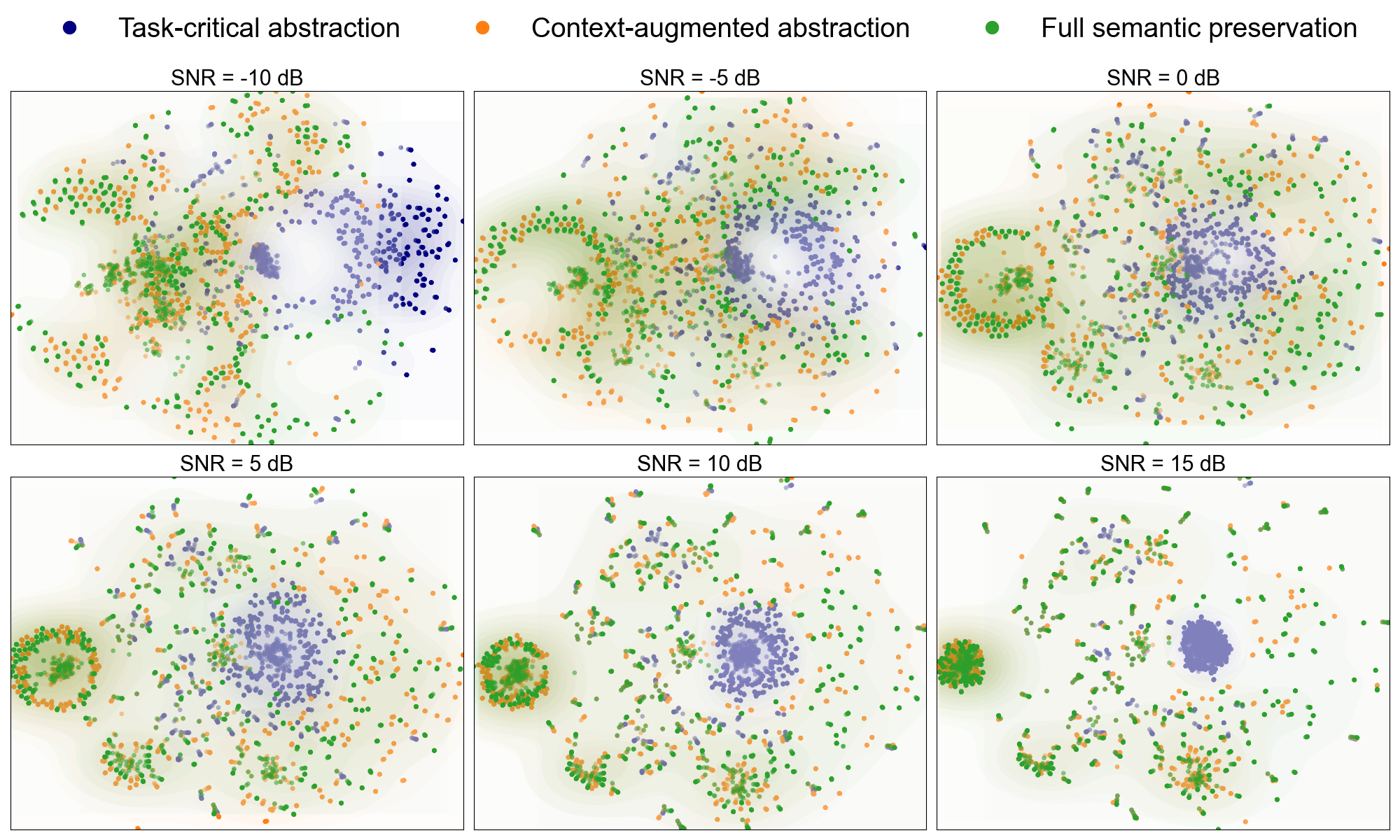}
		\caption{t-SNE visualization of CLIP-based semantic embedding distributions under different semantic abstraction levels and SNR conditions. The projections provide qualitative insights into representation organization.}
		\label{fig:entropy}
	\end{figure*}
	\subsection{Simulation Setup}
    We evaluate the semantic abstraction module using the MS COCO $2017$ dataset \cite{7505636}. Pretrained YOLOv$5$s\footnote{\url{https://github.com/ultralytics/yolov5}} and DeepLabV$3$-ResNet$50$\footnote{\url{https://github.com/pytorch/vision}} models perform object detection and semantic mask generation, respectively, while CLIP ViT-B/$16$ extracts features for t-SNE visualization. 
	For the channel evaluation, we add AWGN at the receiver and vary the received SNR from $-10$ to $21$~dB across the considered experiments. 
	Each packet contains $512$ bits, and the BS discards erroneous packets without requesting retransmission. 
	For the LoS analysis, we set the urban-environment parameters in \eqref{plos} to $9.61$ and $0.16$, respectively \cite{li2025covert}. We consider elevation angles of $10^\circ$, $30^\circ$, and $50^\circ$, and set the conditional LoS and NLoS SNRs to $15$ and $0$~dB, respectively.

	
    Unless otherwise stated, we consider $M=5$ UAVs belonging to $K=2$ types. 
	The type parameters $\theta_1$ and $\theta_2$ are randomly sampled from $[10,50)$ and $[100,200)$, respectively \cite{10515203}, while the type proportions $(\lambda_1,\lambda_2)$ follow a Dirichlet distribution \cite{10638123}. 
	The unit resource cost of covert SemCom service provisions $a$ is randomly sampled within $[80,100)$. The loss aversion $\eta$ and loss distortion $\zeta^{+}, \zeta^{-}$ are set to $0.5$, $1$ and $1$, respectively. The weighting factors $\upsilon$ and $\beta$ are set to $200$ and $50$, respectively.

    RDSAC employs a diffusion-based contract generation network and twin contract quality networks \cite{10638123}, with the main hyperparameters listed in Table~\ref{tab:rdsac_hyperparameters}. The generation network uses a $16$-dimensional sinusoidal diffusion-step embedding and a variance-preserving noise schedule with $b_{\min}=0.1$ and $b_{\max}=10$. We set $T=5$ to limit denoising overhead while retaining iterative contract refinement, and use a lower learning rate for the generation network to keep policy updates conservative. All experiments are implemented in PyTorch with CUDA 12.6 on a desktop computer equipped with an NVIDIA GeForce RTX 3080 GPU.

	\subsection{Performance Evaluation of the Semantic Abstraction Regulation Module}
	\begin{table}[t] 
		\renewcommand{\arraystretch}{1.2}
		\setlength{\heavyrulewidth}{1.3pt} 
		\caption{Comparison of $\mathcal{Q}$ values and PER for different semantic abstraction levels (i.e., $g_1$, $g_2$, and $g_3$) under varying SNR conditions.} 
		\centering \label{tab:effQ}
		\begin{tabular}{ccccc} 
			\toprule
			\cellcolor{gray!5}SNR (dB)
			& \cellcolor{gray!5}$\mathbf{\mathcal{Q}_{g_1}}$
			& \cellcolor{gray!5}$\mathbf{\mathcal{Q}_{g_2}}$
			& \cellcolor{gray!5}$\mathbf{\mathcal{Q}_{g_3}}$
			& \cellcolor{gray!5}PER \\
			\midrule
			$-3$
			& \cellcolor{orange!5}\ensuremath{\mathbf{0.000073}}
			& \cellcolor{orange!5}\ensuremath{0.000027}
			& \cellcolor{orange!5}\ensuremath{0.000007}
			& \cellcolor{blue!5}\ensuremath{0.999992} \\
			
			$0$
			& \cellcolor{orange!5}\ensuremath{0.022468}
			& \cellcolor{orange!5}\ensuremath{0.009444}
			& \cellcolor{orange!5}\ensuremath{\mathbf{0.025215}}
			& \cellcolor{blue!5}\ensuremath{0.998047} \\
			
			$3$
			& \cellcolor{orange!5}\ensuremath{0.515429}
			& \cellcolor{orange!5}\ensuremath{0.271028}
			& \cellcolor{orange!5}\ensuremath{\mathbf{0.749682}}
			& \cellcolor{blue!5}\ensuremath{0.961121} \\
			
			$6$
			& \cellcolor{orange!5}\ensuremath{2.766209}
			& \cellcolor{orange!5}\ensuremath{1.942710}
			& \cellcolor{orange!5}\ensuremath{\mathbf{5.319875}}
			& \cellcolor{blue!5}\ensuremath{0.810911} \\
			
			$9$
			& \cellcolor{orange!5}\ensuremath{6.688143}
			& \cellcolor{orange!5}\ensuremath{6.208409}
			& \cellcolor{orange!5}\ensuremath{\mathbf{16.60004}}
			& \cellcolor{blue!5}\ensuremath{0.570755} \\
			
			$12$
			& \cellcolor{orange!5}\ensuremath{10.43494}
			& \cellcolor{orange!5}\ensuremath{12.68313}
			& \cellcolor{orange!5}\ensuremath{\mathbf{33.00704}}
			& \cellcolor{blue!5}\ensuremath{0.347441} \\
			
			$15$
			& \cellcolor{orange!5}\ensuremath{13.06226}
			& \cellcolor{orange!5}\ensuremath{20.51111}
			& \cellcolor{orange!5}\ensuremath{\mathbf{50.02533}}
			& \cellcolor{blue!5}\ensuremath{0.193230} \\
			
			$18$
			& \cellcolor{orange!5}\ensuremath{14.56031}
			& \cellcolor{orange!5}\ensuremath{28.5039}
			& \cellcolor{orange!5}\ensuremath{\mathbf{65.21053}}
			& \cellcolor{blue!5}\ensuremath{0.102205} \\
			
			$21$
			& \cellcolor{orange!5}\ensuremath{15.31169}
			& \cellcolor{orange!5}\ensuremath{36.02593}
			& \cellcolor{orange!5}\ensuremath{\mathbf{77.442}}
			& \cellcolor{blue!5}\ensuremath{0.052649} \\
			\bottomrule
		\end{tabular}
		\vspace{-0.4cm}
	\end{table}

	\begin{table*}[b]
    \renewcommand{\arraystretch}{1.2}
    \setlength{\heavyrulewidth}{1.3pt}

    \caption{Physical-layer quantities and $\mathcal{Q}$ at the feasible transmit-power boundaries under different covertness tolerances.}
    \centering
    \label{tab:abstraction_covertness}

    \begin{tabular}{cccccccc}
        \toprule
        \cellcolor{gray!5}$\epsilon_c$
        & \cellcolor{gray!5}$\xi^\star$
        & \cellcolor{gray!5}$P_{\mathrm{cov}}^{\max}$ (dBm)
        & \cellcolor{gray!5}SNR (dB)
        & \cellcolor{gray!5}PER
        & \cellcolor{gray!5}$\mathbf{\mathcal{Q}_{g_1}}$
        & \cellcolor{gray!5}$\mathbf{\mathcal{Q}_{g_2}}$
        & \cellcolor{gray!5}$\mathbf{\mathcal{Q}_{g_3}}$ \\
        \midrule

        0.050
        & 0.950
        & $-17.254$
        & 1.126
        & \cellcolor{blue!5}\ensuremath{0.993}
        & \cellcolor{orange!5}\ensuremath{0.091}
        & \cellcolor{orange!5}\ensuremath{0.042}
        & \cellcolor{orange!5}\ensuremath{\mathbf{0.115}} \\

        0.100
        & 0.900
        & $-14.201$
        & 4.179
        & \cellcolor{blue!5}\ensuremath{0.918}
        & \cellcolor{orange!5}\ensuremath{1.131}
        & \cellcolor{orange!5}\ensuremath{0.678}
        & \cellcolor{orange!5}\ensuremath{\mathbf{1.866}} \\

        0.200
        & 0.800
        & $-11.086$
        & 7.295
        & \cellcolor{blue!5}\ensuremath{0.712}
        & \cellcolor{orange!5}\ensuremath{4.333}
        & \cellcolor{orange!5}\ensuremath{3.481}
        & \cellcolor{orange!5}\ensuremath{\mathbf{9.416}} \\

        0.300
        & 0.700
        & $-9.192$
        & 9.189
        & \cellcolor{blue!5}\ensuremath{0.555}
        & \cellcolor{orange!5}\ensuremath{6.941}
        & \cellcolor{orange!5}\ensuremath{6.571}
        & \cellcolor{orange!5}\ensuremath{\mathbf{17.532}} \\

        \bottomrule
    \end{tabular}

\end{table*}

	\begin{table}[t]
		\renewcommand{\arraystretch}{1.2}
		\setlength{\heavyrulewidth}{1.3pt}
		\caption{Covert semantic effectiveness $\mathcal{Q}$ for different abstraction levels and the corresponding PER under varying UAV--BS LoS probabilities.}
		\centering
		\label{tab:losQ}
		\begin{tabular}{cccccc}
			\toprule
			\cellcolor{gray!5}$\vartheta$ ($^\circ$)
			& \cellcolor{gray!5}${\Pr}^{\mathrm{LoS}}_{u,\mathrm{BS}}$
			& \cellcolor{gray!5}$\mathcal{Q}_{g_1}$
			& \cellcolor{gray!5}$\mathcal{Q}_{g_2}$
			& \cellcolor{gray!5}$\mathcal{Q}_{g_3}$
			& \cellcolor{gray!5}$\mathrm{PER}$ \\
			\midrule
			$10$
			& $0.0997$
			& \cellcolor{orange!5}\ensuremath{1.3227}
			& \cellcolor{orange!5}\ensuremath{2.0538}
			& \cellcolor{orange!5}\ensuremath{\mathbf{5.0109}}
			& \cellcolor{blue!5}\ensuremath{0.9178} \\

			$30$
			& $0.7310$
			& \cellcolor{orange!5}\ensuremath{9.5543}
			& \cellcolor{orange!5}\ensuremath{14.9957}
			& \cellcolor{orange!5}\ensuremath{\mathbf{36.5743}}
			& \cellcolor{blue!5}\ensuremath{0.4097} \\

			$50$
			& $0.9852$
			& \cellcolor{orange!5}\ensuremath{12.8695}
			& \cellcolor{orange!5}\ensuremath{20.2081}
			& \cellcolor{orange!5}\ensuremath{\mathbf{49.2863}}
			& \cellcolor{blue!5}\ensuremath{0.2051} \\
			\bottomrule
		\end{tabular}
	\end{table}

	\begin{figure}[t]
		\centering
		\begin{minipage}[t]{0.45\linewidth}  
			\centering
			\includegraphics[width=\linewidth]{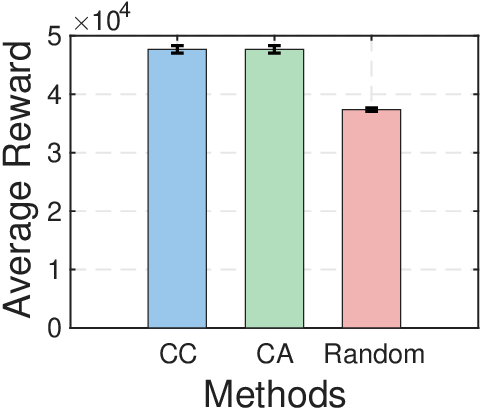}
			\captionsetup{font=footnotesize}
			\caption{Average reward comparison of our scheme with other schemes under PT,  where the reference point $U_\mathrm{ref}=160$.}
			\label{compare}
		\end{minipage}
		\hspace{0.02\linewidth}  
		\begin{minipage}[t]{0.45\linewidth}
			\centering
			\includegraphics[width=\linewidth]{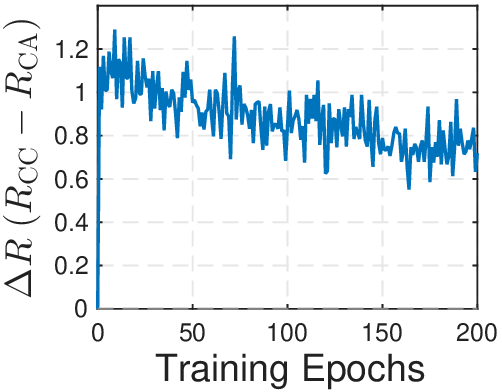}
			\captionsetup{font=footnotesize}
			\caption{Difference in average rewards between the complete-information contract and the asymmetric-information contract during training under PT, where the reference point $U_\mathrm{ref}=160$.}
			\label{chazhi}
		\end{minipage}
	\end{figure}



    In this subsection, we evaluate the effectiveness of the proposed semantic abstraction regulation module by jointly examining the semantic representation organization and the corresponding covert semantic effectiveness $\mathcal{Q}$ under varying channel conditions. 
	For a given visual scene, three hierarchical semantic abstraction levels are constructed, including task-critical abstraction ($g=1$), context-augmented abstraction ($g=2$), and full semantic preservation ($g=3$). 
	The perception models, i.e., YOLOv$5$ and DeepLabv$3$, identify task-relevant regions and provide masks for constructing the abstraction levels before encoding.
	For representation analysis, the Contrastive Language-Image Pre-training (CLIP) visual encoder \cite{shen2021much} maps the observations at all abstraction levels into a common embedding space, enabling a consistent comparison of their feature distributions.
    We then project the high-dimensional semantic embeddings into a two-dimensional space using t-distributed Stochastic Neighbor Embedding (t-SNE) \cite{10294259} to visualize the representation organization patterns induced by different semantic abstraction levels. 
	It should be noted that t-SNE is employed only for qualitative visualization rather than quantitative entropy estimation. 
	Therefore, the resulting low-dimensional projections are used only to provide intuitive insights into the representation diversity and local semantic organization under different abstraction levels.
	In addition, to investigate the impact of transmission impairments on semantic representations, the above analysis is conducted under multiple SNR settings. Beyond feature-space characterization, we further report the corresponding covert semantic effectiveness values $\mathcal{Q}$ achieved by different abstraction levels under the same SNR conditions.

	Figure \ref{fig:entropy} presents the t-SNE visualizations of the semantic embedding distributions corresponding to the three abstraction levels under different channel conditions. 
	When only task-critical abstraction is preserved ($g_1$), the projected semantic features exhibit relatively compact clustering patterns. 
	In comparison, context-augmented abstraction ($g_2$) leads to more diverse structural organization in the projected space, whereas full semantic preservation ($g_3$) produces the richest representation patterns among the considered abstraction levels. 
	These structural differences persist across the evaluated SNR settings.
	Moreover, the projected feature organization becomes less compact as the channel condition deteriorates. 
	These observations indicate that different abstraction levels induce distinguishable semantic representation patterns under varying channel conditions, which is qualitatively consistent with the intended semantic abstraction regulation mechanism. 
	In particular, finer-grained abstraction levels tend to preserve richer semantic content and exhibit higher representation diversity under noisy transmission conditions.
	
	Beyond the qualitative feature-space analysis, Table~\ref{tab:effQ} further reports the covert semantic effectiveness $\mathcal{Q}$ and PER under various SNRs for different abstraction levels. 
	As the SNR increases, the PER consistently decreases, indicating more reliable packet delivery. 
	Correspondingly, $\mathcal{Q}$ increases for all abstraction levels, which reflects improved effective covert semantic service quality under favorable channel conditions. 
	For task-critical abstraction ($g_1$), $\mathcal{Q}$ varies relatively smoothly across different SNRs, indicating stable but limited semantic service quality. 
	For context-augmented abstraction ($g_2$), $\mathcal{Q}$ achieves a better balance between semantic preservation and transmission robustness. 
	For full semantic preservation ($g_3$), the richest semantic content leads to the highest $\mathcal{Q}$ under high-SNR conditions, but its performance becomes more sensitive when the channel quality degrades. 
	These results illustrate the dependence of reconstruction-based semantic effectiveness on both the abstraction level and channel quality.

    Table~\ref{tab:abstraction_covertness} summarizes the covert-feasible transmit-power boundaries and the corresponding SNR, PER, and covert semantic effectiveness $\mathcal{Q}$ under different covertness tolerances.
	These boundaries represent the maximum transmit powers permitted by the covertness constraint under the specified detection and channel conditions.
	As $\epsilon_c$ increases from $0.05$ to $0.3$, the covertness constraint is relaxed, allowing the maximum feasible transmit power to increase from $-17.254$ to $-9.192$~dBm.
	The corresponding SNR increases from $1.126$ to $9.189$~dB, while the PER decreases from $0.993$ to $0.555$.
	Consequently, $\mathcal{Q}$ increases for all three abstraction levels, with $g_3$ achieving the highest reconstruction-based $\mathcal{Q}$ throughout the considered tolerance range.
	For a given tolerance, all abstraction levels use identical physical-layer transmission settings.
	Under the adopted model, Willie's observation distributions are therefore independent of $g$, yielding identical detection performance across the levels.
	Thus, the differences in $\mathcal{Q}$ reflect abstraction-dependent reconstruction quality under the same covert-feasible link conditions.
	
	To further examine the impact of propagation conditions, Table~\ref{tab:losQ} evaluates the sensitivity of semantic performance under varying UAV--BS LoS probabilities. 
	In the adopted two-state channel model, the $15$~dB and $0$~dB operating points in Table~\ref{tab:effQ} are adopted as representative conditional LoS and NLoS states, respectively. 
	As $p_{\mathrm{LoS}}$ increases from $0.0997$ to $0.9852$, the average PER decreases from $0.9178$ to $0.2051$. 
	Meanwhile, the average $\mathcal{Q}$ increases from $1.3227$ to $12.8695$ for $g_1$, from $2.0538$ to $20.2081$ for $g_2$, and from $5.0109$ to $49.2863$ for $g_3$. 
	These improvements result from the higher probability of favorable LoS conditions.
	Although absolute performance varies with LoS availability, $g_3$ consistently achieves the highest average reconstruction-based $\mathcal{Q}$.
	This consistent ranking indicates the robustness of its relative performance advantage to variations in LoS availability within the considered two-state channel model.

	\subsection{Performance Evaluation of the RDSAC Under PT}
	\begin{figure}[!t]
		\centering
		\captionsetup{font=footnotesize}
		\includegraphics[width=0.67\linewidth]{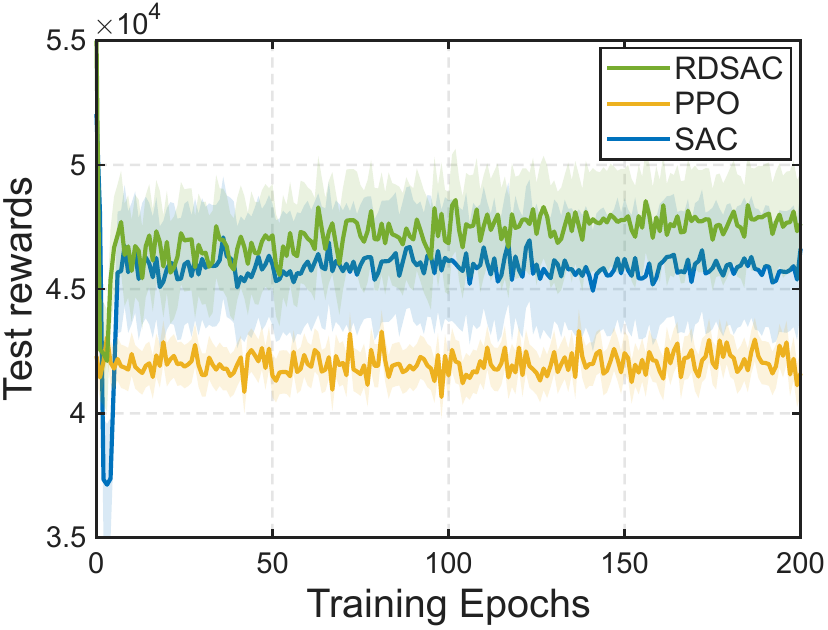}
		\caption{Performance comparison between the RDSAC algorithm and traditional DRL algorithms in optimal contract design under PT, where the reference point $U_\mathrm{ref}=160$.}
		\label{algo}
	\end{figure}
	\begin{figure}[!t]
		\centering
		\captionsetup{font=footnotesize}
		\includegraphics[width=0.67\linewidth]{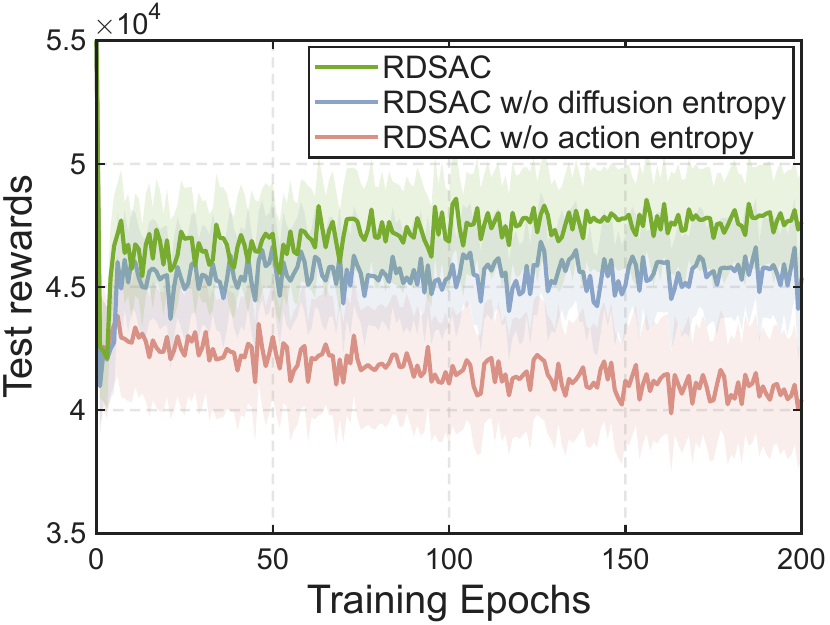}
		\caption{Test reward comparison of RDSAC and its variants without diffusion entropy regularization or action entropy regularization for contract design under PT.}
		\label{fig:entropy_ablation}
	\end{figure}
	\begin{figure}[t]
		\centering
		\captionsetup{font=footnotesize}
		\includegraphics[width=0.8\linewidth]{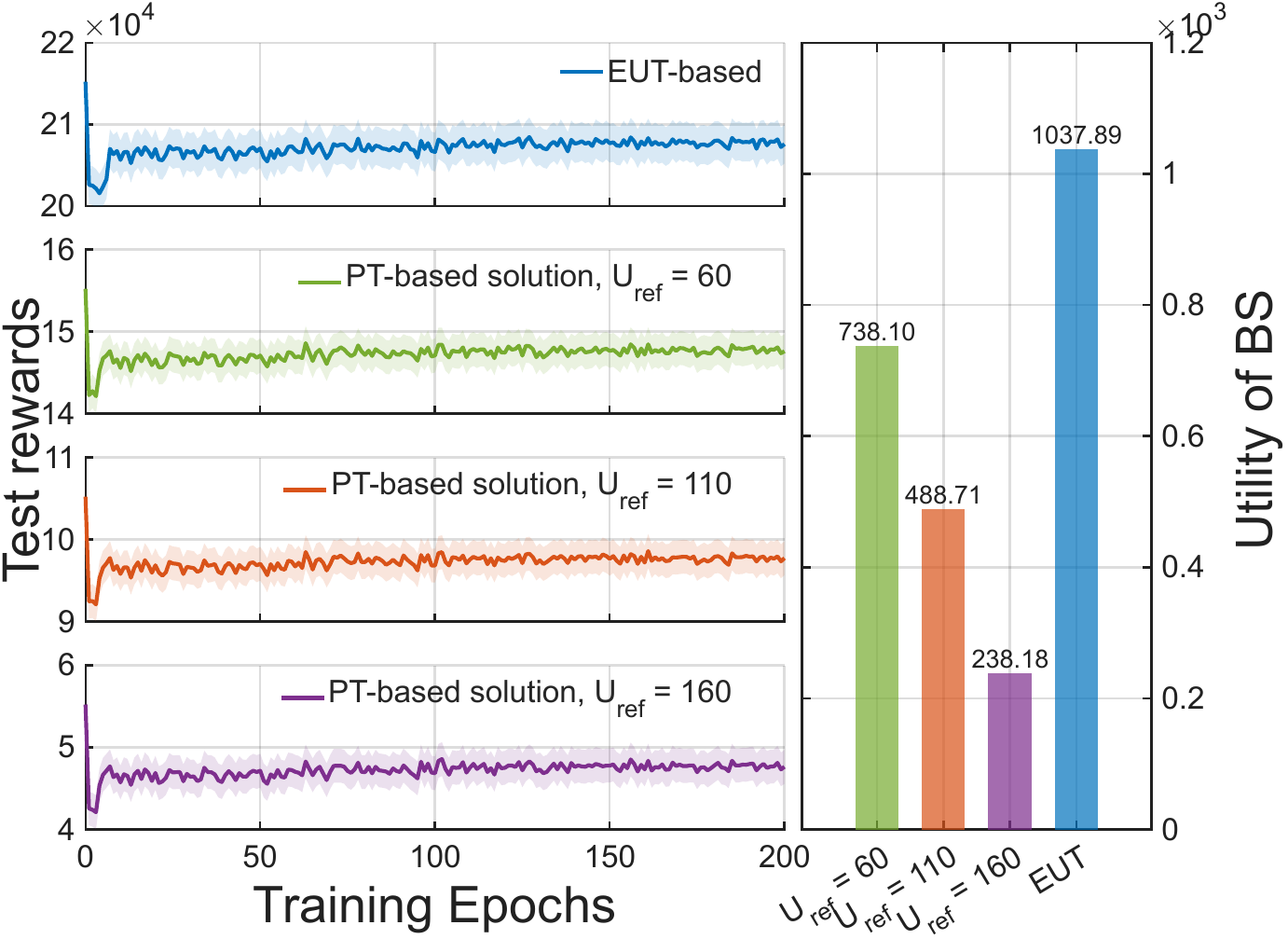}
		\caption{Test reward comparison of the proposed scheme under different reference points $U_\mathrm{ref}$.}
		\label{reference}
	\end{figure}
	\begin{figure*}[t]
		\centering
		\captionsetup{font=footnotesize}
		\subfigure[Environmental state 1]{
			\includegraphics[width=0.2\textwidth]{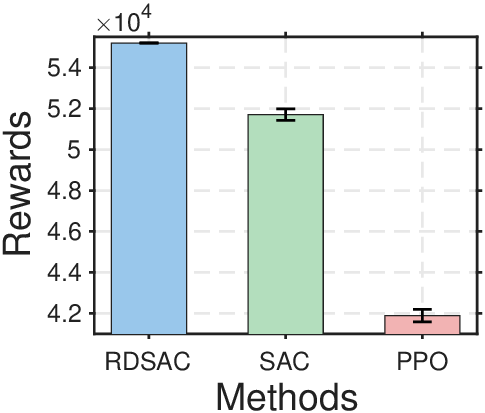}
		}
		\subfigure[Environmental state 2]{
			\includegraphics[width=0.2\textwidth]{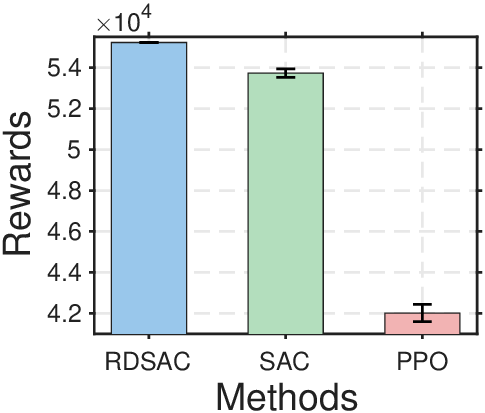}
		}
		\subfigure[Environmental state 3]{
			\includegraphics[width=0.2\textwidth]{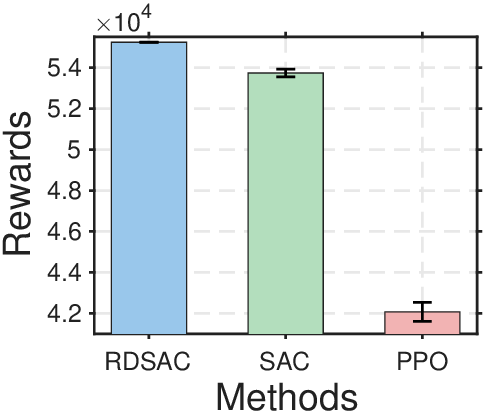}
		}
		\subfigure[Environmental state 4]{
			\includegraphics[width=0.2\textwidth]{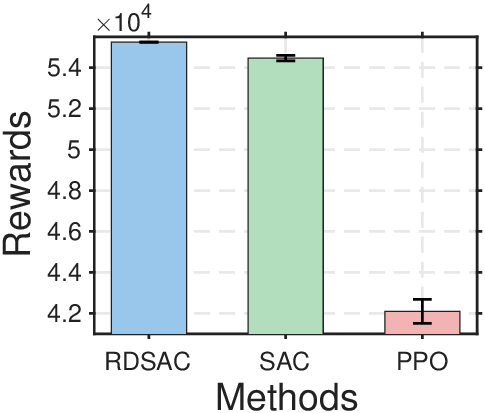}
		}
		\caption{Average testing rewards and standard deviations of different algorithms under diverse environmental settings.}
		\label{fig:four_subfig}
	\end{figure*}
	
	In this subsection, we evaluate the proposed RDSAC algorithm from multiple perspectives, including information asymmetry, algorithmic comparison, prospect-theoretic effects, and environmental robustness. All results are reported in terms of the average test reward and its standard deviation to reflect both performance and stability.
	
	To evaluate the performance of the proposed scheme, we compare the proposed \underline{C}ontract-based incentive mechanism with \underline{A}symmetric information (CA) with other schemes under PT: 1) \textit{\underline{C}ontract-based incentive mechanism with \underline{C}omplete information (CC)} \cite{11341901} that the BS knows the private information of UAVs (i.e., UAV types); 2) \textit{Random policy under asymmetric information} that the BS randomly designs contracts without considering the types of UAVs. Figure~\ref{compare} shows the converged average rewards of different schemes. Both CA and CC significantly outperform the random policy. 
	This result demonstrates the necessity of a contract-based incentive mechanism under asymmetric information. 
	In contrast, random contract generation fails to effectively align the incentives between the BS and UAVs. This leads to much worse performance. Moreover, as shown in Fig.~\ref{chazhi}, the performance gap between CA and CC is marginal. Under identical system settings, CA achieves an average reward after convergence that is only slightly lower than that of CC. This finding indicates that the proposed contract generation model under asymmetric information can closely approach the performance upper bound of the complete-information benchmark. The limited gap is due to the fact that complete information allows the BS to select the exact optimal contract for each UAV. Under asymmetric information, contract design relies on incentive-compatible self-selection, which introduces only minor inefficiency. Although the complete-information scenario yields slightly higher utility, it is impractical in realistic UAV-assisted networks. In practice, UAV information is private and may be misreported \cite{10254627}. Overall, the proposed scheme achieves near-optimal performance under asymmetric information. This makes it more practical and reliable than the idealized complete-information benchmark.
	
	In Fig.~\ref{algo}, we evaluate the performance of optimal contract design under PT. We compare RDSAC with two representative DRL algorithms, namely SAC and PPO. It is observed that RDSAC consistently achieves the highest test rewards among all algorithms. This performance gain is mainly attributed to the diffusion-based strategy generation process, which effectively mitigates the impact of noise and stochasticity during policy learning. 
	In addition, RDSAC incorporates diffusion entropy regularization to provide stable imitation learning signals, thereby preventing policy collapse, as well as action entropy regularization to encourage broader exploration and avoid convergence to suboptimal local optima \cite{wen2025hybridrag}. Overall, these mechanisms jointly enhance learning stability and exploration efficiency, demonstrating the effectiveness of the proposed RDSAC algorithm.

    To assess the individual contributions of the two regularization terms, Fig.~\ref{fig:entropy_ablation} compares RDSAC with two variants obtained by separately removing diffusion entropy regularization and action entropy regularization.
	Over the final $20$ training epochs, RDSAC improves the average test reward by $4.63\%$ and $17.10\%$ over the variants without diffusion entropy regularization and action entropy regularization, respectively.
	Without diffusion entropy regularization, the test reward remains at a lower level, consistent with the benefit of denoising guidance for contract generation.
	Without action entropy regularization, the test reward gradually decreases during the later training stage. This trend suggests premature policy concentration on a narrow set of contracts. Such concentration makes subsequent policy updates more sensitive to value estimation errors, which can degrade the quality of generated contracts. Action entropy regularization maintains action diversity and encourages continued exploration, helping prevent this performance deterioration.
	By incorporating both terms, RDSAC maintains higher rewards during later training.
	
	Figure~\ref{reference} illustrates the impact of different reference points on the proposed contract scheme under PT, together with the EUT-based solution. The PT-based solutions are evaluated under three reference points, i.e., $U_{\mathrm{ref}} \in \{60, 110, 160\}$, while the EUT-based solution corresponds to the proposed contract model without PT. We see that the EUT-based solution consistently achieves higher testing rewards during training and attains the highest BS utility after convergence, since it assumes a risk-neutral BS utility evaluation and ignores the loss-sensitive perception introduced by PT \cite{10638123}.
	In contrast, the performance of PT-based solutions is strongly influenced by the reference point. A lower reference point leads to better training performance and higher BS utility, which decreases from $738.1$ to $488.71$ and $238.18$ as $U_{\mathrm{ref}}$ increases from $60$ to $110$ and $160$, respectively. This trend indicates that higher reference points amplify perceived losses of the BS under PT, thereby reducing their participation incentives and degrading the BS utility. 

    Figure.~\ref{fig:four_subfig} compares the trained RDSAC, SAC, and PPO policies under four environmental configurations.
	Specifically, the parameter tuples $(\theta_1,\theta_2,\mathrm{PER})$ are $(15,120,0.2)$, $(25,140,0.4)$, $(35,160,0.6)$, and $(45,180,0.8)$ for environments 1--4, respectively.
	We evaluate each policy over $50$ episodes per configuration and report the mean and standard deviation of the cumulative episode reward.
	RDSAC consistently achieves the highest mean reward with very small standard deviations, whereas SAC and PPO yield lower rewards with greater variability.
	These results demonstrate the robustness of RDSAC across the tested configurations and its effectiveness in improving system utility while maintaining stable and reliable performance.

	\begin{figure}[t]
		\vspace{-1em}
		\centering
		\captionsetup{font=footnotesize}
		\includegraphics[width=0.67\linewidth]{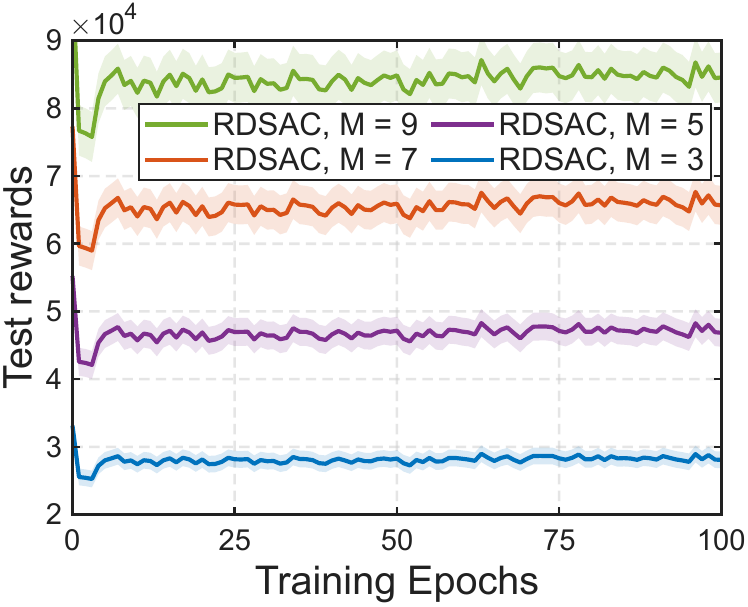}
		\caption{Performance comparison of the RDSAC algorithm in optimal contract design under PT for different numbers of UAVs, where $M\in\{3,5,7,9\}$.}
		\label{fig:uav_count}
	\end{figure}

	\begin{figure}[t]
		\vspace{-1em}
		\centering
		\captionsetup{font=footnotesize}
		\includegraphics[width=0.7\linewidth]{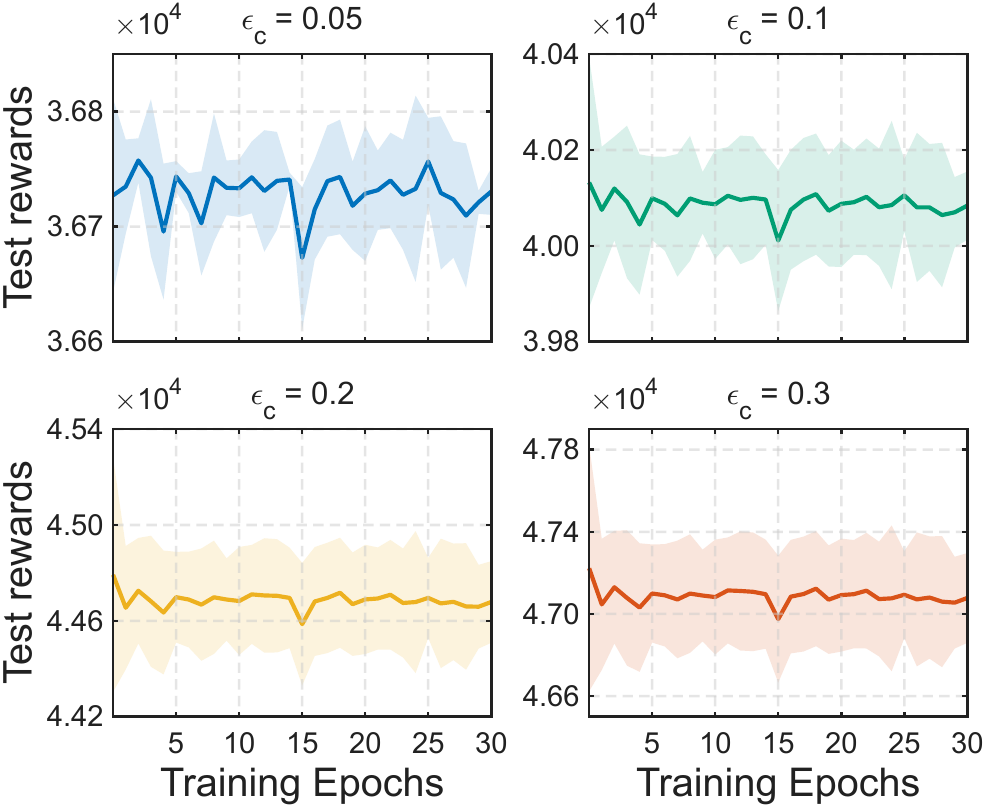}
		\caption{Testing performance of RDSAC under different covertness tolerances. Solid lines and shaded regions indicate the mean and one standard deviation across three independently trained runs, respectively.}
	\label{fig:covertness_performance}
	\end{figure}

    Figure~\ref{fig:uav_count} evaluates the performance of RDSAC under different UAV population sizes, with $M\in\{3,5,7,9\}$.
	The test reward increases approximately proportionally with the number of UAVs.
	This behavior is consistent with the BS utility formulation, in which $M$ scales the aggregate utility for fixed type proportions and per-UAV operating conditions.
	Meanwhile, the four reward trajectories exhibit similar training patterns and fluctuate around their respective reward levels during the later training stage.
	Under the adopted contract model, increasing the UAV population therefore primarily changes the aggregate reward scale, while the observed training behavior remains broadly consistent.

	Finally, to examine the impact of covertness constraints on algorithm performance, 
	Fig.~\ref{fig:covertness_performance} evaluates the testing performance of RDSAC under different covertness tolerances, with $\epsilon_c\in\{0.05,0.1,0.2,0.3\}$.
	For each tolerance, RDSAC is trained independently using three random seeds under the corresponding covert-feasible transmit-power limit.
	The results show that relaxing the covertness constraint leads to higher test rewards.
	A larger $\epsilon_c$ permits higher transmit power, thereby improving the SNR and reducing the PER.
	The resulting improvement in effective semantic quality increases the achievable contract utility.
	Conversely, a stricter covertness constraint limits the feasible transmit power and reduces the attainable reward, highlighting the tradeoff between covertness and system utility.
	
	\section{Conclusion} \label{seven}
    In this paper, we have studied covert SemCom in LAENets under stringent low-probability-of-detection constraints. 
	Specifically, we have proposed an incentive-aware semantic abstraction framework for differentiated covert semantic services under fixed-dimensional transmission payloads.
	The abstraction levels regulate visual content before encoding, while the contract mechanism coordinates semantic quality requirements and rewards across heterogeneous UAVs.
	Moreover, we have designed a covert semantic effectiveness metric to jointly characterize covert feasibility, packet delivery reliability, and abstraction-dependent semantic reconstruction quality.
	To address information asymmetry between the BS and heterogeneous UAVs, we have formulated a contract-theoretic model with PT-based subjective utility evaluation for the BS.
	We have further developed the RDSAC algorithm with diffusion entropy and action entropy regularization to solve the resulting non-convex contract optimization problem. 
	Numerical results have demonstrated the effectiveness of the proposed framework and the performance gains of RDSAC over baseline algorithms.
	For future work, we will study scalable incentive mechanisms for larger UAV networks and theoretical guarantees for risk-sensitive contract optimization.

%
	
	\bibliographystyle{IEEEtran}
	\bibliography{ref}
	
\end{document}